\documentclass[10pt,english,a4paper]{article}

\usepackage{ifpdf}
\usepackage{multido}
\usepackage{ifthen}
\usepackage{color}

\ifpdf
\usepackage[pdftex]{graphicx}
\else
\usepackage[dvips]{graphicx}
\fi

\usepackage{float,floatflt,afterpage}
\usepackage{wrapfig}
\usepackage{placeins}

\usepackage{etex}
\usepackage{tikz}
\usetikzlibrary{snakes}

\usepackage[arrow,matrix,curve]{xy}
\usepackage{subfigure}

\graphicspath{{./figures/}}

\usepackage{babel}
\usepackage[latin1]{inputenc}

\usepackage[centertags,intlimits]{amsmath}
\usepackage{amsthm}
\usepackage{amssymb}
\usepackage{mathtools}
\usepackage{mathrsfs}

\theoremstyle{break}
\allowdisplaybreaks[1]

\usepackage{exscale}
\usepackage{dsfont}
\usepackage{bbm}

\usepackage[numbers,square,sort&compress]{natbib}

\ifpdf
\usepackage[pdftex,
  bookmarks,
  bookmarksopen=true,
  bookmarksnumbered=true,
  backref,
  linkcolor=blue,
  urlcolor=blue,
  pdfauthor={Andreas Fischle},
  pdftitle={The geometrically nonlinear Cosserat micropolar shear-stretch
    energy. Part II: Non-classical energy-minimizing microrotations in
    3D and their computational validation}
  hyperfigures=true,
  pdfpagelabels,
  pagebackref,
  hypertexnames=true,
  plainpages=false,
  naturalnames]{hyperref}\else
\usepackage[dvips,breaklinks]{hyperref}                                         \usepackage{breakurl}
\fi

\usepackage{hypernat}

\usepackage{url} \usepackage{breakurl} 

\makeatletter
\def\url@leostyle{  \@ifundefined{selectfont}{\def\UrlFont{\sf}}{\def\UrlFont{\small\ttfamily}}}
\makeatother
\usepackage{enumerate}
\usepackage{verbatim}        \usepackage{lscape}          \usepackage{array,longtable} \usepackage{listings}        \usepackage{caption}

\usepackage{makeidx}

                                                                                                                                                                                                                                                                                                            \def\and{{\rm and}}            \def\det#1{{\rm det}\,{#1}}            
                                                  
\def\skew{{\rm skew}}                                  
\newfont{\Sf}{cmssbx10 scaled 2074}
\newbox{\assem}
\savebox{\assem}(15,2)[bl]{
\begin{normalsize}
\unitlength=1.0mm
\begin{picture}(6.00,0.00)
\put(.0,1.0){\makebox(0,0)[cc]{$\cup$}}
\put(.0,-1.0){\makebox(0,0)[cc]{$\mbox{\scriptsize \it n=1}$}}
\put(.0,4.0){\makebox(0,0)[cc] {\scriptsize $\it n_{\mbox{\tiny\it elm}}$}}
\end{picture}
\end{normalsize}}

\newbox{\asse}
\savebox{\asse}(15,2)[bl]{
\begin{normalsize}
\unitlength=1.0mm
\begin{picture}(8.00,0.00)
\put(.0,2.0){\makebox(0,0)[cc]{\Sf A}}
\end{picture}
\end{normalsize}}

                                                        \def\sqtwo3{{\textstyle {\sqrt{2 \over 3}}}}   \newcommand{\IP}{{\rm I\kern-.18em P}}           \newcommand{\II}{{\rm I\kern-.18em I}}           \newcommand{\IF}{{\rm I\kern-.25em F}}           \newcommand{\IE}{{\rm I\kern-.25em E}}           \def\IR{{\rm I\kern-.15em R}}

\newcommand{\ia}{{\rm\kern.24em                     \vrule width.02em height0.9ex depth-.05ex
   \kern-.26em a}}
\newcommand{\ic}{{\rm\kern.24em                     \vrule width.02em height0.9ex depth-.05ex
   \kern-.26em c}}
\newcommand{\IA}{{\rm\kern.22em                      \vrule width.02em
        height0.5ex depth 0ex
    \kern-.24em A}}
\newcommand{\IC}{{\rm\kern.24em                     \vrule width.02em height1.4ex depth-.05ex
   \kern-.26em C}}

\DeclareMathOperator{\dist}{dist}

\newcommand{\Dphi}{\nabla\! \mathbf{\varphi}}

\renewcommand{\epsilon}{\varepsilon}

\newcommand{\dV}{\,{\rm dV} }

\newcommand{\norm}[1]{\|#1\|}
\newcommand{\abs}[1]{\left| #1 \right|}

\newtheorem{lem}{Lemma}[section]

\newtheorem{rem}[lem]{Remark}
\newtheorem{defi}[lem]{Definition}

\newtheorem{theo}[lem]{Theorem}

\newtheorem{cor}[lem]{Corollary}

\newtheorem{prob}[lem]{Problem}
\newtheorem{compres}[lem]{Computer Assisted Result}
\newtheorem{compval}[lem]{Computational Validation}

\newcommand{\leref}[1]{Lemma \ref{#1}}
\newcommand{\theref}[1]{Theorem \ref{#1}}

\newcommand{\coref}[1]{Corollary \ref{#1}}

\newcommand{\remref}[1]{Remark \ref{#1}}
\newcommand{\deref}[1]{Definition \ref{#1}}

\newcommand{\probref}[1]{Problem \ref{#1}}

\newcommand{\R}{\mathbb{R}}

\newcommand{\C}{\mathbb{C}}

\renewcommand{\S}{\mathbb{S}}
\DeclareMathOperator{\SL}{SL}
\DeclareMathOperator{\GL}{GL}
\DeclareMathOperator{\SO}{SO}

\DeclareMathOperator{\skewop}{skew}
\renewcommand{\skew}{\skewop}
\DeclareMathOperator{\diag}{diag}
\DeclareMathOperator{\sym}{sym}
\DeclareMathOperator{\Tr}{tr}

\DeclareMathOperator{\so}{\mathfrak{so}}

\DeclareMathOperator{\polar}{R_{\rm p}}

\newcommand{\Sym}{ {\rm{Sym}} }
\newcommand{\Psym}{ {\rm{PSym}} }

\newcommand{\id}{{\boldsymbol{\mathbbm{1}}}}
 
\DeclareMathOperator{\Det}{det}

\renewcommand{\det}[1]{ {\Det[{#1}]} }
\newcommand{\tr}[1]{ {\Tr \left[{#1}\right]} }

\newcommand{\secref}[1]{Section \ref{#1}}
\newcommand{\figref}[1]{Figure \ref{#1}}

\definecolor{orange}{rgb}{1.0,0.5,0}

\DeclareMathOperator{\Reals}{\mathbb{R}}
\renewcommand{\R}{\Reals}

\DeclareMathOperator{\Quaternions}{\mathbb{H}}
\renewcommand{\H}{\Quaternions}

\DeclareMathOperator{\argminmathop}{arg\,min}
\newcommand{\argmin}[2]{\mathchoice{\underset{#1}{\argminmathop}\, {#2}}{\argminmathop_{#1}\, {#2}}{}{}}

\newcommand{\mathematica}{\texttt{Mathematica}}

\newcommand{\restrict}[2]{\left.#1\right\rvert_{#2}\,}
\newcommand{\scalprod}[2]{\big<#1,\,#2\big>}

\newcommand{\setdef}[2]{\lbrace #1 \;\vert\; #2\rbrace}

\DeclareMathOperator{\spanop}{span}

\newcommand{\vspan}[1]{\spanop\left(\left\{ #1 \right\}\right)}

\newcommand{\hsnorm}[1]{\left\lVert #1 \right\rVert}

\DeclareMathOperator{\RPosZ}{\sideset{}{_0^+}\Reals}
\DeclareMathOperator{\RPos}{\Reals^+}

\newcommand{\eqdef}{\coloneqq}
\newcommand{\defeq}{\eqqcolon}
\newcommand{\eqiso}{\cong}
\newcommand{\isequivto}{\Longleftrightarrow}

\newcommand{\mrot}{R}
\newcommand{\mstretch}{\overline{U}}

\DeclareMathOperator{\rpolar}{rpolar}

\newcommand{\ump}{u^{\rm mmp}}
\newcommand{\smp}{s^{\rm mmp}}

\DeclareMathOperator{\sradmm}{\rho_{\mu,\,\mu_c}}

\DeclareMathOperator{\wmm}{W_{\mu,\mu_c}}
\DeclareMathOperator{\wsym}{W_{1,0}}

\DeclareMathOperator{\wmmred}{W^{\rm red}_{\mu,\mu_c}}
\DeclareMathOperator{\wsymred}{W^{\rm red}_{1,0}}

\newcommand{\domc}{\mathrm{Dom}^\mathrm{C}}
\newcommand{\domn}{\mathrm{Dom}^\mathrm{NC}}

\newcommand{\countres}{
  \setcounter{equation}{0}
  \setcounter{figure}{0}
  \setcounter{table}{0}
}

                     \renewcommand{\baselinestretch}{1.0}          

\makeatletter

\renewcommand{\itemize}{  \ifnum \@itemdepth >\thr@@\@toodeep\else
    \advance\@itemdepth\@ne
    \edef\@itemitem{labelitem\romannumeral\the\@itemdepth}    \expandafter
    \list
      \csname\@itemitem\endcsname
      {\def\makelabel##1{\hss\llap{##1}}        \topsep=.8ex\itemsep=-.2ex}  \fi}

\renewcommand\section{\@startsection {section}{1}{\z@}  {-3.5ex \@plus -1ex \@minus -.2ex}  {2.3ex \@plus.2ex}  {\boldmath\normalfont\Large\bfseries}}
\renewcommand\subsection{\@startsection{subsection}{2}{\z@}  {-3.25ex\@plus -1ex \@minus -.2ex}  {1.5ex \@plus .2ex}  {\boldmath\normalfont\large\bfseries}}
\renewcommand\subsubsection{\@startsection{subsubsection}{3}{\z@}  {-3.25ex\@plus -1ex \@minus -.2ex}  {1.5ex \@plus .2ex}  {\boldmath\normalfont\normalsize\bfseries}}
\renewcommand\paragraph{\@startsection{paragraph}{4}{\z@}  {3.25ex \@plus1ex \@minus.2ex}  {-1em}  {\boldmath\normalfont\normalsize\bfseries}}
\renewcommand\subparagraph{\@startsection{subparagraph}{5}{\parindent}  {3.25ex \@plus1ex \@minus .2ex}  {-1em}  {\boldmath\normalfont\normalsize\bfseries}}

\makeatother
 
\title{\vspace{-2.5cm}The geometrically nonlinear Cosserat micropolar shear-stretch
       energy. \mbox{Part II:}
       Non-classical energy-minimizing microrotations in 3D and their
       computational validation}

\author{Andreas Fischle
\!\!\thanks{Corresponding author: Andreas Fischle,
Institut f\"ur Numerische Mathematik,
TU Dresden,
Zellescher Weg 12-14,
01069 Dresden,
Germany,
email: andreas.fischle@tu-dresden.de}
\quad and \quad
Patrizio Neff
\!\!\thanks{Patrizio Neff,
Head of Lehrstuhl f\"{u}r Nichtlineare Analysis und Modellierung,
Fakult\"{a}t f\"{u}r Mathematik,
Universit\"{a}t Duisburg-Essen,
Thea-Leymann Str. 9,
45127 Essen,
Germany,
email: patrizio.neff@uni-due.de}}

\begin{document}
\selectfont
\maketitle

\pagenumbering{arabic}

\vspace{-0.75cm}
\begin{center}
  \textit{Dedicated to \textsc{Gianfranco Capriz} on the occasion of his 90th
    birthday:\\
    -- we bow in great admiration to his lifetime achievement! --}
\end{center}

\begin{center}\textbf{Abstract}\end{center}
\begin{center}
  \begin{minipage}{0.95\textwidth}
    In any geometrically nonlinear, isotropic and quadratic Cosserat
    micropolar extended continuum model formulated in the deformation
    gradient field $F \eqdef \nabla\varphi: \Omega \to \GL^+(n)$
    and the microrotation field $\mrot: \Omega \to \SO(n)$,
    the shear--stretch energy is necessarily of the form
    \begin{align*}
      \wmm(\mrot\,;F)& \eqdef
      \mu  \hsnorm{\sym(\mrot^TF - \id)}^2
      + \mu_c\hsnorm{\skew(\mrot^TF - \id)}^2\;.
    \end{align*}
    We aim at the derivation of closed form expressions for the
    minimizers of $\wmm(\mrot\,;F)$ in $\SO(3)$, i.e., for
    the set of optimal Cosserat microrotations in dimension $n\!=\!3$,
    as a function of $F \in \GL^+(n)$. In a previous contribution (Part I),
    we have first shown that, for all $n \geq 2$, the full range of weights
    $\mu > 0$ and $\mu_c \geq 0$ can be reduced to either a classical or a
    non-classical limit case. We have then derived the associated closed
    form expressions for the optimal planar rotations in $\SO(2)$ and
    proved their global optimality. In the present contribution (Part II),
    we characterize the non-classical optimal rotations in dimension
    \mbox{$n\!=\!3$}. After a lift of the minimization problem to the unit
    quaternions, the Euler--Lagrange equations can be symbolically solved
    by the computer algebra system \mathematica. Among the symbolic expressions
    for the critical points, we single out two candidates
    $\rpolar^\pm_{\mu,\mu_c}(F) \in \SO(3)$ which we analyze and for which
    we can computationally validate their global optimality by Monte
    Carlo statistical sampling of $\SO(3)$. Geometrically, our proposed
    optimal Cosserat rotations $\rpolar^\pm_{\mu,\mu_c}(F)$ act in the
    ``plane of maximal strain'' and our previously obtained explicit
    formulae for planar optimal Cosserat rotations in $\SO(2)$ reveal
    themselves as a simple special case. Further, we derive the
    associated reduced energy levels of the Cosserat shear--stretch
    energy and criteria for the existence of non-classical optimal
    rotations.
\end{minipage}
\end{center}

\vspace*{0.125cm}
{\small

   }
 \countres

\setcounter{tocdepth}{1}
\renewcommand{\baselinestretch}{-1.0}\normalsize
{
  \small
  \tableofcontents
}
\renewcommand{\baselinestretch}{1.0}\normalsize

\section{Introduction}\label{sec:intro}
In this second part (Part II) of a series, we consider the weighted optimality
problem for the Cosserat shear--stretch energy $\wmm:\; \SO(n) \,\times\, \GL^+(n) \to \RPosZ$,\\
\begin{equation}
  \wmm(\mrot\,;F) \;\eqdef\; \mu\, \hsnorm{\sym(\mrot^TF - \id)}^2
  \,+\,
  \mu_c\,\hsnorm{\skew(\mrot^TF - \id)}^2\;.
  \label{eq:intro:wmm}
\end{equation}
The arguments are the deformation gradient field
$F \eqdef \nabla\varphi: \Omega \to \GL^+(n)$ and the microrotation field
$\mrot: \Omega \to \SO(n)$ evaluated at a given point of the domain
$\Omega$. This energy arises in any geometrically nonlinear, isotropic
and quadratic Cosserat micropolar continuum model. Note that it is always
possible to express the local energy contribution in a Cosserat model
as $W = W(\mstretch)$, where $\mstretch \eqdef R^TF$ is the first Cosserat
deformation tensor. This reduction follows from objectivity requirements
and has already been observed by the Cosserat
brothers~\cite[p.~123, eq.~(43)]{Cosserat09},
see also~\cite{Eringen99} and~\cite{Maugin:1998:STPE}.
Since $\mstretch$ is in general non-symmetric, the most general isotropic
and quadratic local energy contribution which is zero at the reference
state is given by
\begin{equation}
  \label{intro:wgeneral}
  \underbrace{\mu\, \hsnorm{\sym(\mstretch - \id)}^2
    \,+\,
    \mu_c\,\hsnorm{\skew(\mstretch - \id)}^2}_{\text{``shear--stretch energy''}}
  \quad+\quad
  \underbrace{\frac{\lambda}{2}\,\tr{\mstretch - \id)}^2}_{\text{``volumetric energy''}}\;.
\end{equation}
The last term will be discarded in the following, since it couples
the rotational and volumetric response, a feature not present in the
well-known isotropic linear Cosserat models.\footnote{The Cosserat brothers never
  proposed any specific expression for the local energy $W = W(\mstretch)$.
  The chosen quadratic ansatz for $W = W(\mstretch)$ is motivated by a direct
  extension of the quadratic energy in the linear theory of Cosserat
  models, see, e.g.~\cite{Jeong:2009:NLIC,Neff_Jeong_bounded_stiffness09,Neff_Jeong_Conformal_ZAMM08}. We consider a true volumetric-isochoric split
  in~\secref{sec:discussion:application}.}

Let us now proceed to the primary objective of our present contribution
\begin{prob}[Weighted optimality in dimension $n = 3$]
  Let $\mu > 0$ and $\mu_c \geq 0$. Compute the set of optimal rotations
  \label{intro:prob_wmm}
  \begin{equation}
    \argmin{\mrot\,\in\,\SO(3)}{\wmm(\mrot\,;F)}
    \;\eqdef\;
    \argmin{\mrot\,\in\,\SO(3)}{      \left\{
      \mu\, \hsnorm{\sym(\mrot^TF - \id)}^2
      \,+\,
      \mu_c\,\hsnorm{\skew(\mrot^TF - \id)}^2\right\}
    }
    \label{eq:intro:weighted_wmm}
  \end{equation}
  for given parameter $F \in \GL^+(3)$ with distinct singular values
  $\sigma_1 > \sigma_2 > \sigma_3 > 0$.
\end{prob}
We use the notation $\sym(X) \eqdef \frac{1}{2}(X + X^T)$,
$\skew(X) \eqdef \frac{1}{2}(X - X^T)$, $\scalprod{X}{Y} \eqdef \tr{X^TY}$
and we denote the induced Frobenius matrix norm by
$\hsnorm{X}^2 \eqdef \scalprod{X}{X} = \sum_{1 \leq i,j \leq n} X_{ij}^2$.
In mechanics applications, the weights $\mu > 0$ and $\mu_c \geq 0$ can
be identified with the Lam\'e shear modulus $\mu > 0$ from linear elasticity
and the so-called Cosserat couple modulus $\mu_c \geq 0$. The
parameter $\lambda$ in the most general form of the energy~\eqref{intro:wgeneral}
can further be identified with the second Lam\'e parameter. Note that
the interpretation of the Cosserat couple modulus $\mu_c$ is somewhat
delicate, see, e.g.,~\cite{Neff_ZAMM05}, which is one of the fundamental
motivations for this second contribution in a series.

In Part I of this paper~\cite{Fischle:2015:OC2D}, we have proved a
still surprising reduction lemma~\cite[Lem.~2.2, p.~4]{Fischle:2015:OC2D}
for the material parameters (weights) $\mu$ and $\mu_c$ which is valid
for \emph{all} space dimensions $n \geq 2$. This lemma singles out a
\emph{classical parameter range} $\mu_c \geq \mu > 0$ and a
\emph{non-classical parameter range} $\mu_c \geq \mu > 0$ for $\mu$
and $\mu_c$ and reduces both ranges to an associated limit case. The
\emph{classical limit case} is given by $(\mu,\mu_c) = (1,1)$ and the
\emph{non-classical limit case} is given by $(\mu,\mu_c) = (1,0)$. We
then apply the parameter reduction~\cite[Lem.~2.2,
  p.~4]{Fischle:2015:OC2D} to~\probref{intro:prob_wmm} and solve it in
dimension $n = 2$. This allows us to discuss the optimal planar
Cosserat rotations and we observe that the classical and the
non-classical parameter ranges for $\mu$ and $\mu_c$ characterize two
substantially different types of optimal Cosserat rotations.

To explain this difference, we first need to introduce the polar
factor $\polar(F) \in \SO(n)$ which is obtained from the right polar
decomposition $F = \polar(F)\,U(F)$ of the deformation gradient $F \in
\GL^+(n)$. Here, $U(F) \eqdef \sqrt{F^TF} \in \Psym(n)$ denotes the
positive definite symmetric right Biot-stretch tensor. We recall
that the eigenvalues of $U \in \Psym(n)$ are by definition the
singular values $\sigma_1 > \sigma_2 > \sigma_3 > 0$ of the
deformation gradient $F \in \GL^+(n)$.

In the classical parameter range $\mu_c \geq \mu > 0$, the polar
factor $\polar$ admits a variational characterization
which is noteworthy in its own right: namely, for all $n \geq 2$, it
is the \emph{unique} minimizer for~\eqref{eq:intro:wmm} as a generalized
version of Grioli's theorem shows,
see~\cite{Grioli40,Guidugli:1980:EPP,Neff_Grioli14},
or~\cite[Cor.~2.4, p.~5]{Fischle:2015:OC2D}. This variational
characterization of the polar factor inspired us to introduce the
following
\begin{defi}[Relaxed polar factor(s)]
  Let $\mu > 0$ and $\mu_c \geq 0$. We denote the set-valued mapping that assigns
  to a given parameter $F \in \GL^+(n)$ its associated set of energy-minimizing
  rotations by
  $$\rpolar_{\mu,\mu_c}(F) \quad \eqdef\quad \argmin{\mrot\,\in\,\SO(n)}{\wmm(\mrot\,;F)}\;.$$
\end{defi}
In dimensions $k = 2,3$, we denote the associated optimal Cosserat rotation
angles by $\alpha_{\mu,\mu_c}(F) \subset (-\pi,\pi]$. More generally, in what
follows, we shall denote the rotation angle of the (absolute) rotation
field $R \in \SO(k)$ by $\alpha \in (-\pi,\pi]$ and the rotation axis by
$r \in \S^{k - 1}$. By $\S^{n -1} \subset \R^{n}$, we denote the unit
$n - 1$-sphere. In dimension $k = 3$, we use the well-known axis-angle
parametrization of rotations which we write as $[\alpha,\, r^T]$.

Since the classical parameter domain $\mu_c \geq \mu > 0$ is very well
understood by now, we can focus on the non-classical parameter range
$\mu > \mu_c \geq 0$ in our efforts to solve~\probref{intro:prob_wmm}.
Here, the parameter reduction lemma~\cite[Lem.~2.2, p.~4]{Fischle:2015:OC2D}
allows us to restrict our attention to the non-classical
limit case $(\mu,\mu_c) = (1,0)$, because it shows the equivalence
\begin{equation}
  \argmin{\mrot\,\in\,\SO(n)}{W_{\mu,\mu_c}(\mrot\,;F)}
  \quad=\quad
  \argmin{\mrot\,\in\,\SO(n)}{W_{1,0}(\mrot\,; \widetilde{F}_{\mu,\mu_c})}
\end{equation}
for all $n \geq 2$. On the right hand side appears the
\emph{rescaled deformation gradient}
$\widetilde{F}_{\mu,\mu_c} \eqdef \lambda^{-1}_{\mu,\mu_c} \cdot F \in \GL^+(n)$
which is obtained from $F \in \GL^+(n)$ by multiplication with the inverse of
the \emph{induced scaling parameter}
$\lambda_{\mu,\mu_c} \eqdef \frac{\mu}{\mu - \mu_c} > 0$. We note that we use
the previous notation throughout the text and further introduce
the \emph{singular radius} $\rho_{\mu,\mu_c} \eqdef \frac{2\mu}{\mu - \mu_c}$.

It follows that the set of optimal Cosserat rotations can be described by
\begin{equation}
  \rpolar_{\mu,\mu_c}(F) \;=\; \rpolar_{1,0}(\widetilde{F}_{\mu,\mu_c})
\end{equation}
for the entire non-classical parameter range $\mu > \mu_c \geq 0$.
This simplifies our main objective~\probref{intro:prob_wmm} considerably,
since it suffices now to solve
\begin{prob}[Weighted optimality in the non-classical limit
    case $(\mu,\mu_c) = (1,0)$]
  Let $\mu > 0$ and $\mu_c \geq 0$. Compute the set of optimal rotations
  \label{intro:prob_wsym}
  \begin{equation}
    \argmin{\mrot\,\in\,\SO(3)}{\wsym(\mrot\,;F)}
    \;\eqdef\;
    \argmin{\mrot\,\in\,\SO(3)}{\hsnorm{\sym(\mrot^TF - \id)}^2}
    \label{eq:intro:weighted_wsym}
  \end{equation}
  for given parameter $F \in \GL^+(3)$ with distinct singular values
  $\sigma_1 > \sigma_2 > \sigma_3 > 0$.
\end{prob}

Regarding our~\probref{intro:prob_wsym} at hand, we will see
in~\secref{sec:minimization} that there are in general two
energy-minimizing solutions with a certain symmetry. They both have the
same rotation axis but differ by the sign of their respective rotation
angles which allows us to select the corresponding branches by that sign.
Accordingly, we introduce the notations $\rpolar^{\pm}_{\mu,\mu_c}(F)$
and $\alpha^\pm_{\mu,\mu_c}(F)$. Loosely spoken, we will see that
the optimal Cosserat rotations
coincide with the polar factor $\polar$ in a certain compressive regime
for $F \in \GL^+(3)$, but deviate in a certain expansive regime.
We shall precisely characterize this in terms of the singular radius
$\sradmm$. Such a material behavior is commonly referred to as
a \emph{tension-compression asymmetry} which is an interesting natural
phenomenon studied in the material sciences, see,
e.g.,~\cite{Gall:1999:TCA,Gall:1999:RTTCA} and~\cite{Sehitoglu:2000:CRN}
for experimental studies of nickel titanium (NiTi) shape memory single
crystals for a glimpse on this broad subject.\footnote{We do not
  claim that such materials can actually be realistically modelled
  as a Cosserat continuum, although it is not impossible.}

\probref{intro:prob_wsym} is a minimization problem on the matrix
Lie group $\SO(3)$ of special orthogonal matrices parameterized by
the deformation gradient $F \in \GL^+(3)$ in the identity component
of the general linear group. Although it is not our duty, we want
to point to some valuable introductory references to this subject.
An excellent general reference for minimization problems on manifolds
is the text by Absil, Mahony and Sepulchre~\cite{Absil:2009:OAMM}.
There, also numerical solution approaches are presented. For an
introduction to Lie groups and matrix groups, we refer to,
e.g.,~\cite{Baker:2012:MG,JMLee02} and~\cite{KHNeeb91}.
For compact Lie groups and their representation theory, see,
e.g.,~\cite{Hofmann:2006:SCG} and~\cite{Broecker:2003:RCLG}.
There is also a growing body of closely related work treating minimization
problems on matrix groups and Grioli's theorem in a similar
context~\cite{Neff:2014:LMP,Neff_Grioli14,Lankeit:2014:MML,Neff:2009:SSNC}.

Instead of turning towards the solution of~\probref{intro:prob_wsym}
right away, we take a step back and notice that there is still another
opportunity for simplification which reduces the space of parameters
$F \in \GL^+(3)$ to the space of ordered singular values
$\sigma_1 \geq \sigma_2 \geq \sigma_3 > 0$ of $F$.
This can be achieved by a principal axis transformation which introduces
a relative rotation $\widehat{R}$ and allows us to introduce
\begin{defi}[Cosserat shear--stretch energy for the relative rotation $\widehat{R}$]
  \label{defi:wrel}
  Let $\mu > 0$, $\mu_c \geq 0$ and let $D \eqdef \diag (\sigma_1, \sigma_2, \sigma_3)$
  with $\sigma_1 > \sigma_2 > \sigma_3 > 0$ the singular values of $F \in \GL^+(3)$.
  The \emph{energy of the relative rotation} $\widehat{\mrot} \in \SO(3)$ is given
  by
  \begin{equation}
    \widehat{W}_{\mu,\mu_c}(\widehat{\mrot}\,;D)
    \;\eqdef\; W_{\mu,\mu_c}(\widehat{\mrot}^T\,;D)
    \;\eqdef\; \mu\;\hsnorm{\sym(\widehat{\mrot}D - \id)}^2 + \mu_c\;\hsnorm{\skew(\widehat{\mrot}D - \id)}^2
    \;.
  \end{equation}
\end{defi}
This transformation is described in~\secref{sec:minimization} and leads us
to the reduced
\begin{prob}[Optimality of relative rotations in dimension $n = 3$]
Let $\mu = 1$ and $\mu_c = 0$. Compute the set of optimal relative rotations
\label{prob:relative_rhat}
\begin{equation}
  \argmin{\widehat{\mrot}\,\in\,\SO(3)}{\widehat{W}_{1,0}(\widehat{\mrot}\,;D)}
  \quad = \quad
  \argmin{\widehat{\mrot}\,\in\,\SO(3)}{\hsnorm{\sym(\widehat{\mrot}\diag(\sigma_1,\sigma_2,\sigma_3) - \id)}^2}
  \label{prob:min}
\end{equation}
for a given diagonal matrix $D \eqdef \diag(\sigma_1,\sigma_2,\sigma_3)$
with $\sigma_1 > \sigma_2 > \sigma_3 > 0$ the ordered singular values of
the deformation gradient $F \in \GL^+(3)$.
\end{prob}

In this text, we strive to mark quantities related to \emph{relative}
rotations with a ``hat''-symbol, e.g., we write $\widehat{R} \in \SO(3)$.
Further, we note that although, for now, we explicitly exclude the case of
multiple singular values of $F$ from our analysis, there is no major
obstruction. The technical treatment would, however, clutter our exposition
of the basic mechanisms which we want to distill here.

At present, an explicit formal solution for the three-dimensional
problem (let alone the $n$-dimensional problem) seems out of reach for
us. We have, however, successfully computed explicit formulae for the
critical points of the Cosserat shear--stretch energy by the use of
computer algebra from which we have determined optimal solutions.
For this approach to succeed, we first lift the Cosserat shear--stretch
energy expressed in principal axis
coordinates to the sphere of unit quaternions $\S^3 \subset \H$ and
subsequently apply the Lagrange
multiplier technique for minimization with
equality constraints, see, e.g.,~\cite{Hestenes:1975:OTF}. The unit
quaternions form a two-sheeted cover of $\SO(3)$ and allow for a
convenient representation of rotations in three-space. For a preceding
successful application of quaternions to represent the rotational
degress of freedom in Cosserat theory, see,
e.g.,~\cite{Muench07_diss}. A highly interesting recent approach
to Cosserat shell theory which also uses quaternions is based on
geodesic finite elements, see~\cite{Sander:2014:NGNCS}
and~\cite{Grohs:2013:ODEGDF,Sander:2012:GFESG}.

This paper is now structured as follows: in~\secref{sec:minimization},
we introduce the lift of the Cosserat shear--stretch energy from $\SO(3)$
to the sphere of unit quaternions $\S^3 \subset \H \eqiso \R^4$. We then
state the corresponding Euler--Lagrange equations and present the
energy-minimizing solutions. The complete set of critical points computed
by \mathematica ~\cite{Mathematica10} is provided in
Appendix~\ref{sec:appendix}.
In~\secref{sec:discussion}, we present a geometric interpretation of the
optimal Cosserat rotations $\rpolar^\pm_{\mu,\mu_c}(F)$ in terms of the maximal
mean planar stretch $\ump$ for the entire non-classical parameter range
$\mu > \mu_c \geq 0$. This leads us to introduce a classical and a non-classical
domain for the parameter $F \in \GL^+(3)$ for which we also derive some
interesting alternative criteria. This illuminates the bifurcation behavior of
$\rpolar_{\mu,\mu_c}(F)$. Further, we compute the associated reduced energy
levels $W^{\rm red}_{\mu,\mu_c}(F)$ for the Cosserat shear--stretch energy.
Then in~\secref{sec:validation}, we shed light on our methodology for the
analysis of the critical points and the experimental computational validation
of the energy-minimizing Cosserat rotations using statistical (Monte Carlo)
methods. Finally, we summarize our findings in a short conclusion presented
in~\secref{sec:conclusion}.
 \countres
\section{Solvable Euler-Lagrange equations: transformation, lift and Lagrange multipliers}
\label{sec:minimization}

In this section, we use a classical result from the representation
theory of compact Lie groups to cast the reduced minimization
problem stated as~\probref{prob:relative_rhat} into a form which
allows us to symbolically compute explicit expressions for the
critical points using~\mathematica.

It is well-known that the Lie group of unit quaternions
$\S^3 \subset \H \eqiso \R^4$
is closely related to the matrix group of rotations
$\SO(3)$, see, e.g.,~\cite{Marsden:2013:IMS} or~\cite[Chap.~9]{Gallier:2015:NDGL}. More precisely, the unit quaternions $\S^3$ form a double cover
of the matrix group $\SO(3)$. For a general introduction to analysis
on smooth manifolds which includes smooth coverings, see, e.g.,~\cite{JMLee02}.
For a dynamical systems approach to quaternions, see, e.g.,~\cite{Novelia:2015:GSO}
which nicely demonstrates the usefulness of quaternions for mechanics
applications with constraints. A more algebraic approach to quaternions
with historical remarks is given in~\cite{Ebbinghaus:1990:Numbers},
and, finally, for those who enjoy the classics,
see~\cite{Hamilton:1843:ANS} and~\cite{Hamilton:1866:EOQ}.

\subsection{Transformation into principal directions}
In order to reduce the parameter space $\GL^+(3)$, we use the (unique)
polar decomposition\footnote{For an introduction to the polar
  and singular value decomposition, see, e.g.,~\cite{DSerre02} and for
  recent related results on variational characterizations of the polar
factor $\polar(F)$, see~\cite{Neff_Grioli14,Lankeit:2014:MML,Neff:2014:LMP}
and references therein.}
\mbox{$F = \polar U$} and the (non-unique) spectral
decomposition of $U =\sqrt{F^TF} \in \Psym(3)$ given by
$U = QDQ^T$, $Q \in \SO(3)$, and expand
\begin{equation}
  R^TF = R^T\polar U = R^T\polar QDQ^T\;.
\end{equation}
Here, the diagonal matrix $D = \diag(\sigma_1,\sigma_2,\sigma_3)$
contains the eigenvalues of $U$ on its diagonal. These are precisely the singular
values of $F \in \GL^+(3)$. In fact, this is a particular form of the
singular value decomposition (SVD). If $F$ has only simple singular values,
then it is always possible to choose the rotation $Q$ such that an ordering
$\sigma_1 > \sigma_2 > \sigma_3 > 0$ is achieved.

Exploiting that $Q \in \SO(3)$, it is now possible to carry out a transformation
of the Cosserat shear--stretch energy into principal axis coordinates --
essentially due to isotropy of the energy. For the actual computation,
note first that
\begin{align}
  &Q^T(\sym(R^TF) - \id)Q = Q^T\left(\sym(R^T\polar QDQ^T) - \id)\right)Q\notag\\
  &\quad=\sym(Q^TR^T\polar QDQ^TQ - Q^TQ) = \sym(\underbrace{Q^TR^T\polar Q}_{\defeq\;\widehat{R}}D - \id) = \sym(\widehat{R}D - \id)\;.\label{eq:symtransform}
\end{align}
In the process, it is natural to introduce the rotation
\begin{equation}
  \label{eq:Rhat}
  \widehat{R} \eqdef Q^TR^T\polar Q
\end{equation}
which acts \emph{relative} to the polar factor $\polar $ in the coordinate
system induced by the columns of $Q$, i.e., in a positively oriented frame
of principal directions of $U$. This interpretation is also nicely illustrated
by the inverse formula
\begin{equation}
  \label{eq:R}
  R = \left(Q\widehat{R}Q^T\polar ^T\right)^T = \polar Q\widehat{R}^TQ^T
\end{equation}
which allows to recover the original absolute rotation $R$ from the relative
rotation $\widehat{R}$. Our next step is to insert the transformed symmetric
part $\eqref{eq:symtransform}$ into the definition of
\begin{equation}
  W_{1,0}(R\,;F) =
  \hsnorm{\sym(\mrot^TF - \id)}^2 =
  \hsnorm{Q^T\sym(\mrot^TF - \id)Q}^2 =
  \hsnorm{\sym(\widehat{R}D - \id)}^2\label{eq:wsymtransformed}\;,
\end{equation}
where we have used that the conjugation by $Q^T$ preserves the
Frobenius matrix norm.

This is a promising simplification of the Cosserat shear--stretch energy,
because it reduces the dimension of the parameter space from $\dim\GL^+(3) = 9$
to only $3$ parameters. However, we still have to account for the non-uniqueness
of $Q$. To this end, we introduce the following symmetric rotation matrices
\begin{equation*}
  Q_1 \eqdef \id, \quad
  Q_2 \eqdef \diag( 1, -1, -1), \quad
  Q_3 \eqdef \diag(-1,  1, -1), \quad
  Q_4 \eqdef \diag(-1, -1,  1)\;,
\end{equation*}
and collect them in a set
$\mathcal{S} \eqdef \left\{Q_1, Q_2, Q_3, Q_4\right\} \subset \SO(3)$.
This set forms a discrete subgroup of $\SO(3)$ which is isomorphic to the Klein
four-group $K_4 \eqiso \mathbb{Z}^2 \times \mathbb{Z}^2$, as is easily inferred
by a comparison of the multiplication tables.

\begin{rem}[Uniqueness of the factor $Q$]
  Let $\sigma_1 > \sigma_2 > \sigma_3 > 0$ be the ordered eigenvalues
  of $U \eqdef \sqrt{F^TF}$ and let $D \eqdef \diag(\sigma_1,\sigma_2,\sigma_3)$.
  It is well-known that the factor $Q \in \SO(3)$ in the spectral decomposition
  $U = QDQ^T$ is only determined up to the choice of a right handed orientation
  of the uniquely determined orthogonal eigenspaces of $U$. This corresponds
  to the products $QS$, $S \in \mathcal{S}$, which represent all of these
  possibilities.
\end{rem}

It is easy to see that for any possible choice of right handed orientation
encoded by $S \in \mathcal{S}$, we obtain the same energy level
\begin{align}
  \hsnorm{\sym(\widehat{R}D - \id)}^2
  \;=\; \hsnorm{\sym(\widehat{R}SDS^T - \id)}^2
  \;=\; \hsnorm{\sym(S^T\widehat{R}SD - \id)}^2\
\end{align}
which implies
\begin{equation}
  \argmin{\widehat{R}\,\in\,\SO(3)}{W_{1,0}(\widehat{R}^T\,;D)}
  \;=\;
  \setdef{\,S^T \widehat{R} S\,}{R \in \argmin{\widehat{R}\,\in\,\SO(3)}{W_{1,0}(\widehat{R}^T\,;D)} \;\text{and}\; S \in \mathcal{S}}\;.
\end{equation}
Thus, $\mathcal{S}$ is a symmetry group of the set of energy-minimizing
rotations which acts by conjugation. The previous analysis reveals that the
non-uniqueness of $Q \in \SO(3)$ is not an issue for the minimization problem,
since all possible choices $QS$, $S \in \mathcal{S}$, lead to the same energy
level.\footnote{A consistent choice of $Q(F) \in \SO(3)$ for \emph{different}
  values of $F \in \GL^+(3)$ is certainly to be advised for the numerical
  computation of a \emph{field} of minimizers $\rpolar^\pm_{\mu,\mu_c}(F(x))$
  depending on $x \in \Omega$. The inversion formula~\eqref{eq:R}
  explicitly depends on the choice of $Q$ and is sensitive to flips of the
  subspace orientation $Q \mapsto QS$, $S \in \mathcal{S}$.}

Without any loss of generality, we may henceforward focus on the solution of
\begin{equation}
  \argmin{\widehat{R} \in \SO(3)}{\hsnorm{\sym(\widehat{R}D - \id)}^2}
  \quad=\quad
  \argmin{\widehat{R} \in \SO(3)}{W_{1,0}(\widehat{R}^T\,;D)}\;.
\end{equation}
This proves the reduction of~\probref{intro:prob_wmm} to the
minimization problem described in~\probref{prob:relative_rhat}
in~\secref{sec:intro} for the non-classical limit case
$(\mu,\mu_c) = (1,0)$. The same principal axes transformation
can also be carried out for arbitrary values of $\mu$ and $\mu_c$
which gives rise to~\deref{defi:wrel}.

In what follows, we denote the rotation angle of the (absolute)
microrotation field $R \in \SO(3)$ by $\alpha \in (-\pi,\pi]$ and
the axis of the rotation by $r \in \S^2$, where $\S^n \subset \R^{n + 1}$
denotes the unit $n$-sphere. This leads us to the axis-angle representation
of a rotation which we write as $[\alpha,\, r^T]$. In what follows,
we work with different parametrizations of the group of rotations $\SO(3)$
simultaneously. Thus, we introduce the symbol $\equiv$ in order to identify
rotations in $\SO(3)$ which are described with respect to different
parametrizations of $\SO(3)$. For example, we might write for the relative
rotation $\widehat{R} \equiv [\hat{\beta},\, (\hat{r}_1, \hat{r}_2, \hat{r}_3)]$
and for a unit quaternion $q \in \S^3$ describing $R \in \SO(3)$,
we have $q \equiv R \equiv -q$. We see that, in general, this
binary relation is \emph{not unique} since the parametrizations
need not be one-to-one.

The symmetry group $\mathcal{S} \eqdef \left\{Q_1, Q_2, Q_3, Q_4\right\}$ hints
at the structure of the set of optimal relative Cosserat rotations. In our
previously introduced notation, we find:
\begin{equation}
  \begin{aligned}
    Q_1^T \widehat{R} Q_1 &\,\equiv\, \left[\hat{\beta},\, ( \hat{r}_1,  \hat{r}_2,  \hat{r}_3)\right]\,,
    \quad&\quad&\quad Q_2^T \widehat{R} Q_2  \,\equiv\, \left[\hat{\beta},\, ( \hat{r}_1, -\hat{r}_2, -\hat{r}_3)\right]\,,\\
    Q_3^T \widehat{R} Q_3 &\,\equiv\, \left[\hat{\beta},\, (-\hat{r}_1,  \hat{r}_2, -\hat{r}_3)\right]\,,
    \quad&\quad&\quad Q_4^T \widehat{R} Q_4  \,\equiv\, \left[\hat{\beta},\, (-\hat{r}_1, -\hat{r}_2,  \hat{r}_3)\right]\;.
  \end{aligned}
\end{equation}
We observe that for rotations about the coordinate axes, i.e., with
$\hat{r} = e_n$, $n = 1,2,3$, the rotation axis $\hat{r}$ is either left
invariant or negated. The latter is equivalent to the negation of the
rotation angle $\hat{\beta}$.

\subsection{Lifting the minimization problem to \texorpdfstring{$\S^3$}{the three-sphere}}
The unit quaternions can be identified with the three-sphere
$\S^3 \eqdef \setdef{q \in \H}{\abs{q} = 1}$ which we shall consider as
a submanifold of the ambient coefficient space $\R^4$ of the
quaternion division ring $\H$. Let us choose the coordinates
$(w,x,y,z) \in \R^4$, i.e., we write quaternions as
$q = w + i x + j y + k z \in \H$.

In order to cast the minimization problem into a form which lends itself
to the derivation of a closed form solution, it is helpful to simplify the
domain of minimization, i.e., to choose a well-adapted system of coordinates.
We achieve this by lifting the Cosserat shear--stretch energy from
$\SO(3)$ to the covering space given by the sphere of unit quaternions
$\S^3 \subset \R^4$. The principal idea is then to extend the covering
map from $\S^3$ to the ambient space $\R^4$ and to apply the Lagrange
multiplier rule with the constraint function $g(q) \eqdef \abs{q}^2 - 1 = 0$.
This approach leads to minimizers in the submanifold of unit quaternions
$q \in \S^3$ which project to energy-minimizing rotations under the well-known
covering homomorphism
\begin{equation}
  \label{eq:pi}
  \pi: \S^3 \to \SO(3),
  \quad\quad
  \pi(q) = \begin{pmatrix}
    1 - 2(y^2 + z^2) &     2(xy - wz)   &      2(xz + wy)  \\
    2(xy + wz)       & 1 - 2(x^2 + z^2) &      2(yz - wx)  \\
    2(xz - wy)       &     2(yz + wx)  & 1 - 2(x^2 + y^2) \\
  \end{pmatrix}\;.
\end{equation}

In order to make our procedure explicit, let us first consider the case
of arbitrary smooth energies $W:\SO(3) \to \R$.
\begin{lem}
Any smooth energy $W:\SO(3) \to \R$ admits a lift to a smooth
energy \mbox{$W^\sharp: \S^3 \to \R$}
\begin{equation*}
\begin{xy}
  \xymatrix @!R {
    \S^3 \ar[rd]^{W^\sharp} \ar[d]^{\pi} & \\
    \SO(3) \ar[r]_W & \R \\
  }
\end{xy}
\end{equation*}
such that minimizers of $W^\sharp$ are projected to minimizers of $W$, i.e.,
\begin{equation}
  \pi(\,\argmin{q\,\in\,\S^3}{W^\sharp(q)}\,)
  \quad=\quad
  \argmin{R\,\in\,\SO(3)}{W(R)}\;.
\end{equation}
\begin{proof}
  The covering map $\pi: \S^3 \to \SO(3)$ defines a surjective Lie group
  homomorphism with $\ker \pi = \{1, -1\}$, see, e.g.,~\cite{Gallier:2015:NDGL}.
  This implies that the unit quaternions form a two-fold
  cover of $\SO(3)$. In particular, the Lie group homomorphism
  $\pi$ is a local diffeomorphism when restricted to a sheet of
  the covering and maps critical unit quaternions in $\S^3$ to critical
  rotations in $\SO(3)$. By definition
  $W^\sharp(q\,;F): \S^3 \times \GL^+(3) \to \RPosZ$,
  $W^\sharp \eqdef W \circ \pi$ is a lift of the Cosserat
  shear--stretch energy to the covering space $\S^3$. Smoothness
  of $W^\sharp$ is obvious since the composition of smooth maps is
  smooth.
\end{proof}
\end{lem}
For any $R \in \SO(3)$ there exists a $q \in \S^3$ which represents
this rotation as $R = \pi(q) \in \SO(3)$. However, this representation
is only unique up to antipodal identification, i.e., $q$ and $-q$
represent the same rotation: $\pi(q) = R = \pi(-q)$. We further note
that $\pi$ can be symbolically evaluated for all $q \in \H$ which
induces an extension.

As previously defined, the covering map $\pi$ is only defined for unit
quaternions $q \in \S^3$. However, in order to apply the Lagrange multiplier
theorem we have to extend it to a suitable neighborhood in the ambient space.
To this end, we introduce the punctured space of non-zero quaternions by
$\mathring{\H} \eqdef \H\setminus{\{0\}} \eqiso \R^4\setminus{\{0\}}$.
Then, identifying $\R^{3\times 3} \eqiso \R^9$ by concatenation of rows,
allows us to consider $\pi: \mathring{\H} \to \R^{3\times 3}$ as a map
$\pi: \mathring{\H} \to \R^9$ which leads us to the following matrix
representation of the derivative
\begin{equation}
  \mathrm{D}_{(w,x,y,z)}\,\pi\left(q(w,x,y,z)\right)\; =\;
  \begin{pmatrix*}[r]
    0    & -2 z & 2 y & 2 z &    0 & -2 x &  -2 y & 2 x &    0\\
    0    &  2 y & 2 z & 2 y & -4 x & -2 w &   2 z & 2 w & -4 x\\
    -4 y &  2 x & 2 w & 2 x &  0   &  2 z &  -2 w & 2 z & -4 y\\
    -4 z & -2 w & 2 x & 2 w & -4 z &  2 y &   2 x & 2 y &    0
  \end{pmatrix*}^T\!.
\end{equation}
It is not hard to infer that $\mathrm{D}_q \pi(q)$ is of rank $4$ for all
$q \in \mathring{\H}$. Hence, the implicit function theorem ensures that
$\pi:\mathring{\H} \to \R^{3\times 3}$ is a local diffeomorphism from the
punctured ambient space $\mathring{\H}$ of the unit sphere $\S^3$ to its
image $\pi(\mathring{\H}) \subset \R^{3\times 3}$.

\begin{defi}[Extension of the lifted energy]
  The extension of the Lie group homomorphism $\pi: \S^3 \to \SO(3)$ to
  $\mathring{\H}$ given by $\pi: \mathring{\H} \to \R^{3\times 3}$ induces
  an \textbf{extension of the lifted energy} to the ambient space $\mathring{\H}$
  \begin{equation}
    W_{1,0}^\sharp:\mathring{\H} \to \R,
    \quad\quad
    W_{1,0}^\sharp(q\,;D) \eqdef \hsnorm{\sym\left(\pi(q)D -\id\right)}^2\;.
  \end{equation}
\end{defi}
Let us abbreviate $\widehat{R}(\hat{q}) \eqdef \restrict{\pi}{\S^3}(\hat{q})$. It is precisely the
restriction of the lifted energy to the unit quaternions for which the Cosserat shear--stretch energy
of the relative rotation is well-defined
$$
\restrict{\widehat{W}_{1,0}^\sharp}{\S^3}(\hat{q}\,;D) \quad=\quad \widehat{W}_{1,0}(\widehat{R}(\hat{q})\,;D)\;.
$$
This extension is simply a mathematical construction, i.e., for $\hat{q} \in \mathring{\H}\setminus\S^3$ the lifted energy $\widehat{W}_{\mu,\mu_c}^\sharp$ \emph{loses its original interpretation as a shear--stretch energy}. Further, we note that
the choice of extension is not unique, but the solutions to the Euler--Lagrange
equations do not depend on the particular extension.\footnote{Alternatively, one may use, e.g., the
following extension which yields pairwise orthogonal columns
\begin{equation*}
  \pi': \mathring{\H} \to \R^{3\times 3}\;,\quad
  \pi'(q)\;\eqdef\;
  \begin{pmatrix}
    w^2+x^2-y^2-z^2 & 2 (x y - w z) & 2 (x z + w y) \\
    2 (x y + w z) & w^2-x^2+y^2-z^2 & 2 (y z - w x) \\
    2 (x z - w y) & 2 (y z + w x) & w^2-x^2-y^2+z^2 \\
  \end{pmatrix}\;.
\end{equation*}
The restrictions $\restrict{\pi}{\S^3} = \restrict{\pi'}{\S^3}$
to the sphere of unit quaternions $\S^3$ are identical.}

\begin{defi}[Lagrange function]
  Consider the constraint function
  $g: \mathring{\H} \to \R$, $g(\hat{q}) \eqdef \abs{\hat{q}}^2 - 1$.
  The \textbf{Lagrange function} for $\widehat{W}_{1,0}^\sharp:\mathring{\H} \to \R$
  is given by
  \begin{equation*}
    \widehat{L}_{1,0}:\mathring{\H} \times \R \to \R,
    \quad\quad
    \widehat{L}_{1,0}(\hat{q},\lambda\,;D)
    \eqdef
    \widehat{W}^\sharp_{1,0}(\hat{q}\,;D) - \lambda\, g(\hat{q})\;.
  \end{equation*}
\end{defi}
Clearly, $g(\hat{q}) = 0$ if and only if $\hat{q} \in \S^3 \subset \mathring{\H}$
which leads us to our final reformulation of the
original~\probref{intro:prob_wmm} in terms of quaternions describing
relative rotations, namely
\begin{prob}[Lagrange multiplier formulation]
  \label{prob:lagrange}
  Compute the critical points of the Lagrange function
  \begin{equation}
    \widehat{L}_{1,0}(\hat{q},\lambda\,;D) \;=\; \hsnorm{\sym\left(\pi(\hat{q})D -\id\right)}^2 - \lambda\, \left(\abs{\hat{q}}^2 - 1\right)
  \end{equation}
  and determine the energy-minimizing solutions.
\end{prob}

The Lagrange function is polynomial. Thus, the application of the
Lagrange multiplier technique leads to an algebraic problem for the
Euler--Lagrange equations which we investigate next.

\subsection{Euler--Lagrange equations, critical points and optimal solutions}

In what follows, a shorthand notation is helpful, so let us introduce
\begin{equation*}
  s_{ij} \eqdef \sigma_i + \sigma_j\quad\text{and}\quad
  d_{ij} \eqdef \sigma_i - \sigma_j\;,\quad i,j = 1,2,3\;.
\end{equation*}

Towards a derivation of the Euler--Lagrange equations in quaternion representation,
we first compute the product
\begin{align*}
\pi(\hat{q}(w,x,y,z))\,D
&= \begin{pmatrix}
    1-2 \left(y^2+z^2\right) & 2 (x y-w z) & 2 (xz + wy) \\
    2 (xy +wz) & 1-2 \left(x^2+z^2\right)  & 2 (yz - wx) \\
    2 (xz -wy) & 2 (yz + wx) & 1-2 \left(x^2+y^2\right)
\end{pmatrix}
\begin{pmatrix}
\sigma_1 & 0 & 0\\
0 & \sigma_2 & 0\\
0 & 0 & \sigma_3\\
\end{pmatrix}\\
&= \begin{pmatrix}
      \sigma_1 \left(1-2 \left(y^2+z^2\right)\right) & 2 \sigma_2 (xy - wz) & 2 \sigma_3 (xz + wy) \\
      2 \sigma_1 (xy + wz) & \sigma_2 \left(1-2 \left(x^2+z^2\right)\right) & 2 \sigma_3 (yz -wx)  \\
      2 \sigma_1 (xz -wy) & 2 \sigma_2 (yz + wx) & \sigma_3 \left(1-2 \left(x^2+y^2\right)\right)  \\
    \end{pmatrix}\;.
\intertext{From this, we infer the symmetric part}
\sym\left(\pi(\hat{q})D\right)
 &= \begin{pmatrix}
    \sigma_1\left(1 - 2 ({y}^2 + {z}^2)\right)& s_{12}\, {x} {y} + d_{12}\, {w} {z} & s_{31}\, {x} {z} + d_{31}\, {w} {y}\\
    s_{12}\, {x} {y}+ d_{12}\, {w} {z} & \sigma_2\left(1 - 2 ({x}^2 + {z}^2)\right) & s_{23}\, {y} {z} + d_{23}\, {w} {x}\\
    s_{31}\, {x} {z} + d_{31}\, {w} {y} & s_{23}\, {y} {z} + d_{23}\, {w} {x} & \sigma_3\left(1 - 2 ({x}^2 + {y}^2)\right)\\
\end{pmatrix}\;.
\end{align*}

Observing that $\sym\left(\pi(\hat{q})D - \id\right) = \sym\left(\pi(\hat{q})D\right) - \id$, we can compute the square of the Frobenius norm. This yields the
following explicit expression for the Lagrange function
$\widehat{L}_{1,0}: \mathring{\H}\times \R \to \R$:
\begin{align*}
\widehat{L}_{1,0}(\hat{q},\lambda\,;D)
&= \left(\sigma_1 (1 \!-\! 2 {y}^2 \!-\! 2 {z}^2) - 1\right)^2
 + \left(\sigma_2 (1 \!-\! 2 {x}^2 \!-\! 2 {z}^2) - 1\right)^2
 + \left(\sigma_3 (1 \!-\! 2 {x}^2 \!-\! 2 {y}^2) - 1\right)^2\\
& \qquad
 + 2\;\left(
   \left(s_{12} {x} {y} + d_{12} {w} {z}\right)^2
 + \left(s_{31} {x} {z} + d_{31} {w} {y}\right)^2
 + \left(s_{23} {y} {z} + d_{23} {w} {x}\right)^2
\right)\\
& \qquad -{\lambda}\; ({w}^2 + {x}^2 + {y}^2 + {z}^2 - 1)\;.
\end{align*}

Let $D = \diag(\sigma_1,\sigma_2,\sigma_3)$ be given. Then a critical tuple of
coefficients $(w,x,y,z,\lambda)$ for the Lagrange function $\widehat{L}_{1,0}$
satisfies the Euler--Lagrange equations in quaternion representation, i.e.,
\begin{equation}
  \mathrm{D}_{(w,x,y,z,\lambda)}\, \widehat{L}_{1,0}\left(\hat{q}(w,x,y,z),\lambda\,;D\right) \;=\; 0\;.
\end{equation}
After a lengthy computation in components (for which we have used
\mathematica), one obtains an explicit form of the Euler--Lagrange
equations for $\widehat{L}_{1,0}$ which is equivalent to the
following parameter-dependent system of polynomials
{\small
  \begin{equation}
    \label{eq:EL_quat}
  \begin{aligned}
0 &= \mathbf{w} \cdot\left(d_{23}^2 \,\mathbf{x}^2 + d_{31}^2 \,\mathbf{y}^2 + d_{12}^2 \,\mathbf{z}^2 -\frac{\,\mathbf{\lambda}}{2}\right)\\
0 &= \mathbf{x} \cdot\left(
  d_{23}^2 \,\mathbf{w}^2
  + 4 (\sigma_2^2 + \sigma_3^2) \,\mathbf{x}^2
  + (4\sigma_3^2 + s_{12}^2) \,\mathbf{y}^2
  + (4\sigma_2^2 + s_{31}^2) \,\mathbf{z}^2
  -\left(d_{23}^2 + (s_{23} - 2)s_{23})\right)
  -\frac{\,\mathbf{\lambda}}{2}\right)\\
0 &= \mathbf{y} \cdot\left(
  d_{31}^2 \,\mathbf{w}^2
  + 4 (\sigma_3^2 + \sigma_1^2 ) \,\mathbf{y}^2
  + (4 \sigma_1^2 + s_{23}^2) \,\mathbf{z}^2
  + (4 \sigma_3^2 + s_{12}^2) \,\mathbf{x}^2
  -\left(d_{31}^2 + (s_{31} - 2)s_{31})\right)
  -\frac{\,\mathbf{\lambda}}{2}\right)\\
0 &= \mathbf{z} \cdot\left(
  d_{12}^2 \,\mathbf{w}^2
  \,+ 4 (\sigma_1^2 + \sigma_2^2) \,\mathbf{z}^2
  + (4 \sigma_2^2 + s_{31}^2) \,\mathbf{x}^2
  + (4 \sigma_1^2 + s_{23}^2) \,\mathbf{y}^2
  -\left(d_{12}^2 + (s_{12} - 2)s_{12})\right)
  -\frac{\,\mathbf{\lambda}}{2}\right)\\
  0 &= \mathbf{w}^2 + \,\mathbf{x}^2 + \,\mathbf{y}^2 + \,\mathbf{z}^2 - 1\;.
  \end{aligned}
\end{equation}}

In general, solution sets of polynomial systems over the field of complex numbers
$\C$ define complex varieties which intuitively can be regarded as
almost-everywhere submanifolds of $\C^n$ with certain singularities.
Real algebraic geometry studies the set of solutions to systems over
real closed fields and the solution sets define so-called semialgebraic
sets~\cite{Bochnak:2013:RAG}. For an exposition of solution methods
for polynomial systems, we refer the interested reader
to~\cite{Sturmfels:2002:SSPE} and~\cite{Cox:2006:UAG}. Note
that in our case both the problem and its solution set are parametrized
by the singular values $\sigma_1 > \sigma_2 > \sigma_3 > 0$ of the
deformation gradient $F \in \GL^+(3)$ encoded by the diagonal matrix
$D = \diag(\sigma_1,\sigma_2,\sigma_3)$. The study of parametrized
polynomial systems is an active research area in computational
algebraic geometry, see, e.g.,~\cite{Montes:2010:GBPP} and references
therein.\footnote{The present authors are not specialists
  in (computational) algebraic geometry. Our goal here is to point out
  some interesting references and developments that might be useful
  for the solution of polynomial systems arising also in other
  applications.}

We briefly introduce the Euler--Lagrange equations obtained by taking
variations on the matrix group $\SO(3)$; cf.~\cite[p.~28]{Neff_Biot07} for
details. Let $\xi = RA \in T_R\SO(3) \eqiso R\cdot\so(3)$ be a direction
in the tangent space at $R \in \SO(3)$. The corresponding
directional derivative of the Cosserat shear--stretch energy
$W_{\mu,\mu_c}(\mrot\,; F)$ is then
\begin{align*}
\label{eq:D1W}
D_{R} W_{\mu,\mu_c}(\mrot\,; F).\xi &=
2\mu\,\scalprod{\sym(\mrot^TF - \id)}{\sym(\xi^TF)}
+ 2\mu_c\,\scalprod{\skew(\mrot^TF)}{\skew(\xi^TF)}\\
&= \scalprod{2\mu\,\sym(\mstretch - \id)
  + 2\mu_c\,\skew(\mstretch)}{A^T\mstretch}\;.
\end{align*}
Equating this derivative with zero and noting, as usual, that this equality must
hold for \emph{all} infinitesimal rotations $A \in \so(3)$,
we obtain the Euler--Lagrange equations in matrix representation.
In particular, any critical
$\mstretch \eqdef R^TF$ must satisfy
\begin{equation}
  \label{eq:EL_SO}
  \skew\left((\mu - \mu_c)\,\mstretch^2 - 2\mu\,\mstretch\right) = 0\;.
\end{equation}
Clearly, the polar factor $\polar$ solves the Euler--Lagrange
equations as it symmetrizes $\mstretch$. Thus, $\polar$ is \emph{always}
a critical point, see, e.g.,~\cite{Bufler85}, or~\cite{Sansour99}.
Under certain conditions on $F$,
however, there may be non-classical critical points and even minimizers
for which $\mstretch$ is no longer symmetric! This observation lies at
the heart of the first collaboration of the present
authors~\cite{Neff_Biot07,Neff_Fischle_GAMM08} and we shall meet this
phenomenon again in the following; cf. also~\cite{Sansour:2008:NCC}.

We have compiled the solution set for the Euler--Lagrange equations in quaternion
representation~\eqref{eq:EL_quat} which we have obtained by using
\mathematica~in Appendix~\ref{sec:appendix}. This permits us to present the
energy-minimizing relative rotations which solve~\probref{prob:lagrange}
without further ado.

\begin{compres}[Energy-minimizing quaternions for $(\mu,\mu_c) = (1,0)$]
  Let $D = \diag(\sigma_1, \sigma_2, \sigma_3)$ with $\sigma_1 > \sigma_2 > \sigma_3 > 0$. Then the quaternion representation of the energy-minimizing
  relative  rotations for $\widehat{W}_{1,0}(\hat{q}\,;D)$ are given by
  the following critical points (listed in Appendix~\ref{sec:appendix}):
  \begin{equation}
    \label{eq:optimal_q}
    \begin{cases}
      \;\hat{q}_{\mathrm{I},1}(D)     \hspace{0.1cm} \quad\equiv\quad \id_3 &,\quad\text{if}\quad s_{12} \eqdef \sigma_1 + \sigma_2 \leq 2\;,\\
      \;\hat{q}^\pm_{\mathrm{II},1}(D) \quad\equiv\quad \left[\pm\arccos(\frac{2}{\sigma_1 + \sigma_2}),\; (0,\,0,\,1)\right] &,\quad\text{if}\quad s_{12} \eqdef \sigma_1 + \sigma_2 \geq 2\;.
    \end{cases}
  \end{equation}
\end{compres}
\textit{Validation.} At present, we cannot give a full proof for this
result. However, we consider our numerical validation to be quite thorough.
For an exposition of our analysis of the critical points compiled in
Appendix~\ref{sec:appendix} and the numerical validation of the presented
energy-minimizing solutions based on extensive random sampling of $\SO(3)$
we refer our reader to~\secref{sec:validation}.

One of the main gaps towards a full proof is the question whether the set of
critical points computed by \mathematica~is complete. Note that our extensive
validation based on random rotations, which exceeds what we can present in a
paper by far, does not hint at the existence of additional critical points.
Solving algebraic problems is the domain where CAS tools such as
\mathematica~do shine brightly.

\begin{cor}[Energy-minimizing relative rotations for $(\mu,\mu_c) = (1,0)$]
 \label{cor:rhat10}
  The solutions to~\probref{prob:relative_rhat} are given by the energy-minimizing
  relative rotations
  {\small
  \begin{equation}
    \widehat{R}_{1,0}^{\pm}(F) \eqdef
    \begin{pmatrix}
      \cos \hat{\beta}^\pm_{1,0}  & -\sin \hat{\beta}^\pm_{1,0} & 0\\
      \sin \hat{\beta}^\pm_{1,0}  &  \cos \hat{\beta}^\pm_{1,0} & 0\\
      0                    &  0                   & 1\\
    \end{pmatrix}
    \stackrel{(\text{if}\; s_{12} \geq 2)}{\vphantom{\Sigma}=}
    \begin{pmatrix}
      \frac{2}{\sigma_1 + \sigma_2}  & \mp \sqrt{1-\left(\frac{2}{\sigma_1 + \sigma_2}\right)^2} & 0\\
      \pm \sqrt{1-\left(\frac{2}{\sigma_1 + \sigma_2}\right)^2} & \frac{2}{\sigma_1 + \sigma_2} & 0 \\
      0 & 0 & 1
    \end{pmatrix}\;.
  \end{equation}
  }
  Here, the optimal relative rotation angles are given by
  \begin{equation}
    \hat{\beta}^\pm_{1,0}(F) \quad\eqdef\quad
    \begin{cases}
      \; 0  &,\;\text{if}\quad s_{12} \eqdef \sigma_1 + \sigma_2 \leq 2\;,\\
      \; \pm\arccos(\frac{2}{\sigma_1 + \sigma_2})
      &,\;\text{if}\quad s_{12} \eqdef \sigma_1 + \sigma_2 \geq 2\;.
    \end{cases}
  \end{equation}
  In particular, for $\sigma_1 + \sigma_2 \leq 2$, we obtain $\widehat{R}_{1,0}^{\pm}(F) = \id$.
\end{cor}

The interpretation of the optimal relative Cosserat rotations is the main
subject of the next section, but in anticipation of this subsequent
discussion we remark that the condition $\sigma_1 + \sigma_2 \leq 2$
characterizes a generalized compressive regime.

 \countres
\section{Optimal Cosserat rotations, maximal mean planar strain and the reduced energy}
\label{sec:discussion}

All proper rotations of euclidean three-space act in a plane perpendicular to
the axis of rotation. From this, a continuum model with rotational degrees of
freedom inherits a certain planar character. In our context, it seems natural
to introduce
\begin{defi}[Maximal mean planar stretch and strain]
  \label{defi:mmpss}
  Let $F \in \GL^+(n)$, $n \geq 2$, with singular values
  $\sigma_1 \geq \sigma_2 \geq \ldots \geq \sigma_n > 0$. We introduce
  the \textbf{maximal mean planar stretch} $\mathbf{\ump}$ and
  the \textbf{maximal mean planar strain} $\mathbf{\smp}$ as follows:
  \begin{equation}
    \begin{aligned}
      \ump(F) &\;\eqdef\; \max_{i \neq j}{\frac{\sigma_i + \sigma_j}{2}} = \frac{\sigma_1 + \sigma_2}{2}\;,\quad\text{and}\\
      \smp(F) &\;\eqdef\; \max_{i \neq j}\frac{(\sigma_i - 1) + (\sigma_j - 1)}{2} = \ump(F) - 1\;.
    \end{aligned}
  \end{equation}
\end{defi}

\begin{defi}[Classical and non-classical domain]
  To any pair of material parameters $(\mu,\mu_c)$ in the non-classical range
  $\mu > \mu_c \geq 0$, we associate the following \textbf{classical domain}
  and \textbf{non-classical domain} for the parameter
  $F \in \GL^+(n)$
  \begin{equation}
    \begin{aligned}
      \domc_{\mu,\mu_c} &\eqdef \setdef{F \in \GL^+(n)}{\smp(\widetilde{F}_{\mu,\mu_c}) \leq 0}\;,
      \quad\text{and}\quad\\
      \domn_{\mu,\mu_c} &\eqdef \setdef{F \in \GL^+(n)}{\smp(\widetilde{F}_{\mu,\mu_c}) \geq 0}\;,
    \end{aligned}
  \end{equation}
  respectively.
\end{defi}

It is straight-forward to derive the following alternative characterizations
{\small
\begin{equation}
  \begin{aligned}
    \domc_{\mu,\mu_c} = \setdef{F \in \GL^+(n)}{\ump(F) \leq \lambda_{\mu,\mu_c}}
    = \setdef{F \in \GL^+(n)}{\sigma_1 + \sigma_2 \leq \sradmm \eqdef \frac{2\mu}{\mu - \mu_c}}\;,\\
    \domn_{\mu,\mu_c} = \setdef{F \in \GL^+(n)}{\ump(F) \geq \lambda_{\mu,\mu_c}}
    = \setdef{F \in \GL^+(n)}{\sigma_1 + \sigma_2 \geq \sradmm \eqdef \frac{2\mu}{\mu - \mu_c}}\;.
  \end{aligned}
\end{equation}
}Note that the intersection $\domc_{\mu,\mu_c} \cap \domn_{\mu,\mu_c} = \setdef{F \in \GL^+(n)}{\smp_{\mu,\mu_c}(F) = 0}$ is not empty. However, the minimizers $\rpolar_{\mu,\mu_c}^\pm(F)$ coincide with the polar factor $\polar(F)$ on this intersection.
This can be seen from the form of the optimal relative rotations in~\coref{cor:rhat10}. In particular, for dimension $n = 3$, we rediscover the following important
characterizations of these domains for the non-classical limit case
$(\mu,\mu_c) = (1,0)$; cf.~\eqref{eq:optimal_q}:
\begin{equation}
\begin{aligned}
  \domc_{1,0} \eqdef \setdef{F \in \GL^+(3)}{s_{12} &\eqdef \sigma_1 + \sigma_2 \leq 2}\;,
  \quad\text{and}\\
  \domn_{1,0} \eqdef \setdef{F \in \GL^+(3)}{s_{12} &\eqdef \sigma_1 + \sigma_2 \geq 2}\;.
\end{aligned}
\end{equation}

Previously, in our~\coref{cor:rhat10}, we have determined the
energy-minimizing relative rotations
\begin{equation}
  \widehat{R}_{1,0}^\pm(D) \;\eqdef\;
  \argmin{\widehat{R}\,\in\,\SO(3)}{\widehat{W}_{1,0}(\widehat{R}\;,D)}
  \;\eqdef\;
  \argmin{\widehat{R}\,\in\,\SO(3)}{\hsnorm{\sym(\widehat{R}D) - \id}^2}\;.
\end{equation}
Let us briefly summarize: for $\ump(F) \leq 1$, i.e., when $F \in \domc_{1,0}$, we have $\widehat{R}_{1,0}^\pm(D) = \id$ which corresponds uniquely to the polar factor $\polar{}$. The minimizers $\rpolar^\pm_{1,0}(F)$ deviate strictly from
$\polar(F)$ for $F \in \domn_{1,0} \setminus \domc_{1,0}$ and are hence
non-classical. Further, expressed in terms of the maximal mean planar
stretch $\ump(F)$, we obtain the alternative representation
\begin{equation}
\widehat{R}_{1,0}^{\pm}(F)
\;=\;
\begin{pmatrix}
  \frac{1}{\ump(F)}               &  \mp\sqrt{1-\frac{1}{\ump(F)^2}} & 0\\
  \pm\sqrt{1-\frac{1}{\ump(F)^2}} &  \frac{1}{\ump(F)} & 0 \\
  0 & 0 & 1
\end{pmatrix}\;.
\end{equation}

\begin{minipage}[t]{0.6\linewidth}
Towards a geometric interpretation of the energy-minimizing Cosserat
rotations $\rpolar^\pm_{1,0}(F)$ in the non-classical limit case
$(\mu,\mu_c) = (1,0)$, we reconsider the spectral decomposition of
$U = QDQ^T$ from the principal axis transformation in~\secref{sec:intro}.
Let us denote the columns of $Q \in \SO(3)$ by $q_i \in \S^2$, $i = 1,2,3$.
Then $q_1$ and $q_2$ are orthonormal eigenvectors of $U$ which correspond
to the largest two singular values $\sigma_1$ and $\sigma_2$ of $F \in \GL^+(3)$.
More generally, we introduce the following

\begin{defi}[Plane of maximal strain]
  \label{defi:pms}
  The \textbf{plane of maximal
    strain} is the linear subspace
  $$\mathrm{P}^{\rm mp}(F) \quad\eqdef\quad \vspan{q_1,q_2} \subset \R^3$$
  spanned by the two maximal eigenvectors $q_1,q_2$ of $U$, i.e.,
  the eigenvectors associated to the two largest singular
  values $\sigma_1 > \sigma_2 > \ldots > \sigma_n$
  of the deformation gradient $F \in \GL^+(n)$, $n \geq 2$.
\end{defi}

We recall that, due to the parameter reduction~\cite[Lem.~2.2]{Fischle:2015:OC2D},
it is always possible to recover the optimal rotations for general
non-classical parameters $\mu > \mu_c \geq 0$
\begin{equation}
  \rpolar_{\mu,\mu_c}(F) \eqdef \argmin{R\,\in\,\SO(3)}{W_{\mu,\mu_c}(\mrot\,;F)}\;.
\end{equation}
from the non-classical limit case $(\mu,\mu_c) = (1,0)$. However, we defer
the explicit procedure for a bit since it is quite instructive
to interpret this distinguished non-classical limit case first.
\end{minipage}
\quad
\begin{minipage}[t]{0.39\linewidth}
  \begin{center}
    {      \setlength{\fboxsep}{0pt}      \setlength{\fboxrule}{1pt}      \mbox{
      \includegraphics[width=\linewidth,clip=true,trim=0.75cm 3cm 0.75cm 2.5cm]{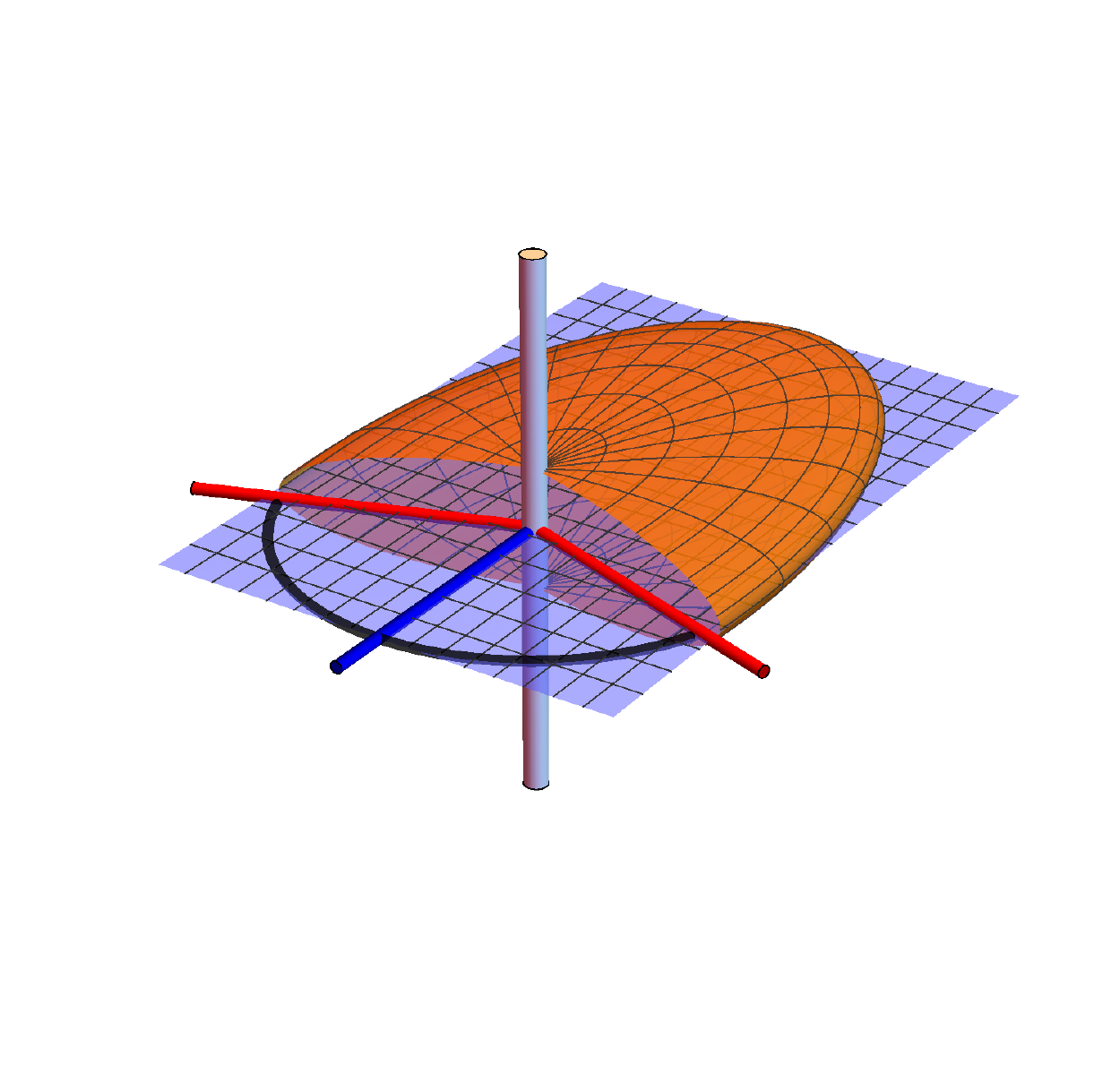}
    }
    }
  \end{center}
  \captionof{figure}{A stretch ellipsoid corresponding
    \label{fig:pms}
    to $(\sigma_1,\sigma_2,\sigma_3) = (4,2,1/2)$. The plane of maximal strain
    $\mathrm{P}^{\rm mp}(F)$ is depicted in blue. The cylinder perpendicular to
    this plane marks the axis of rotation $q_3 \perp \mathrm{P}^{\rm mp}(F)$ of
    $\rpolar^\pm(F)$ which corresponds to the eigenvector associated with the
    smallest singular value $\sigma_3 = 1/2$. The thin blue
    cylinder which bisects the angle enclosed by the opening of the ellipsoid
    corresponds to the polar factor $\polar{}$. Each of the two outer red
    cylinders corresponds to a non-classical minimizer $\rpolar_{1,0}^\pm(F)$.
    The angle enclosed is the optimal relative rotation angle
    $\hat{\beta}_{1,0}^\pm = \pm\arccos(\frac{2}{\sigma_1 + \sigma_2})$.
    This is the major symmetry of the non-classical minimizers.}
\end{minipage}

\begin{rem}[$\rpolar^\pm_{1,0}(F)$ in the classical domain]
  \label{rem:rpolar_class}
  For $\smp(F) \leq 0$ the maximal mean planar stretch is
  non-expansive. By definition, we have $F \in \domc_{1,0}$ in the
  classical domain, for which the energy-minimizing relative rotation
  is given by $\widehat{R}_{1,0}(F) = \id$ and there is no deviation
  from the polar factor. In short $\rpolar_{1,0}^\pm(F) = \polar(F)$.
\end{rem}

Let us now turn to the more interesting non-classical case $F \in \domn_{1,0}$.
\begin{rem}[$\rpolar^\pm_{1,0}(F)$ in the non-classical domain]
  \label{rem:rpolar_nonclass}
  If $F \in \domn_{1,0}$, then by definition $\smp(F) > 0$
  and the maximal mean planar strain is expansive. The deviation of
  the non-classical energy-minimizing rotations $\rpolar^\pm_{1,0}(F)$
  from the polar factor $\polar$ is measured by a rotation
  in the plane of maximal strain $\mathrm{P}^{\rm mp}(F)$ given
  by $\polar(F)^T\rpolar^{\pm}_{1,0}(F) = Q(F)\widehat{R}_{1,0}^\mp(F)Q(F)^T$.
  The rotation axis is the eigenvector $q_3$ associated with the smallest
  singular value $\sigma_3 > 0$ of $F$ and the relative rotation angle
  is given by $\hat{\beta}_{1,0}^\mp(F) = \mp\arccos\left(1/\ump(F)\right)$.
  The rotation angles increase monotonically towards the asymptotic limits
  $$\lim_{\ump(F) \,\to\, \infty} \hat{\beta}_{1,0}^\pm(F) \quad=\quad \pm \pi\;.$$
  In axis-angle representation, we obtain
  \begin{align}
    \widehat{R}_{1,0}^\pm(F) &\quad\equiv\quad \left[\pm \arccos(1/\ump(F)),\, (0,\,0,\,1)\right]\,,\quad\text{and}\\
    \polar^T\rpolar^{\pm}_{1,0}(F) &\quad\equiv\quad \left[\mp \arccos(1/\ump(F)),\, q_3 \right]\;.
  \end{align}
\end{rem}

\begin{cor}[An explicit formula for $\rpolar_{\mu,\mu_c}^\pm(F)$]
  \label{cor:rpolar_formula}
  For the non-classical limit case $(\mu,\mu_c) = (1,0)$ we have
  the following formula for the energy-minimizing Cosserat rotations:
  \begin{equation}
    \rpolar^{\pm}_{1,0}(F)
    \quad\eqdef\quad
    \begin{cases}
      \;\polar(F) &, \text{if}\quad F \in \domc_{1,0}\;,\\
      \;\polar(F)Q(F)\widehat{R}_{1,0}^\mp(F)Q(F)^T &, \text{if}\quad F \in \domn_{1,0}\;.
    \end{cases}
  \end{equation}
  For general values of the weights in the non-classical range
  $\mu > \mu_c \geq 0$, we obtain
  \begin{equation}
    \rpolar^{\pm}_{\mu,\mu_c}(F) \eqdef \rpolar^{\pm}_{1,0}(\widetilde{F}_{\mu,\mu_c})\;,
  \end{equation}
  where $\widetilde{F}_{\mu,\mu_c} \eqdef \lambda^{-1}_{\mu,\mu_c}\,F$ is obtained
  by rescaling the deformation gradient with the inverse of the \emph{induced scaling
  parameter} $\lambda_{\mu,\mu_c} \eqdef \frac{\mu}{\mu - \mu_c} > 0$.
\end{cor}
\begin{proof}
  This is a straightforward application of our equation~\eqref{eq:R}
  which translates relative to absolute rotations derived
  in~\secref{sec:intro} to the optimal relative rotations described
  in~\coref{cor:rhat10}. The second part is non-trivial and follows
  from~\cite[Lem.~2.2]{Fischle:2015:OC2D}.\qedhere
\end{proof}
Note that the previous definition is relative to a fixed choice
of the orthonormal factor $Q(F) \in \SO(3)$ in the spectral
decomposition of $U = QDQ^T$. Further, right from their variational
characterization, one easily deduces that the energy-minimizing
rotations satisfy $\rpolar^\pm_{\mu, \mu_c}(Q\,F) = Q\,\rpolar^\pm_{\mu,\mu_c}(F)$,
for any $Q \in \SO(3)$, i.e., they are objective
functions;~cf.\remref{rem:polar_vs_rpolar}.

\begin{figure}
\begin{center}
  \includegraphics[width=12cm]{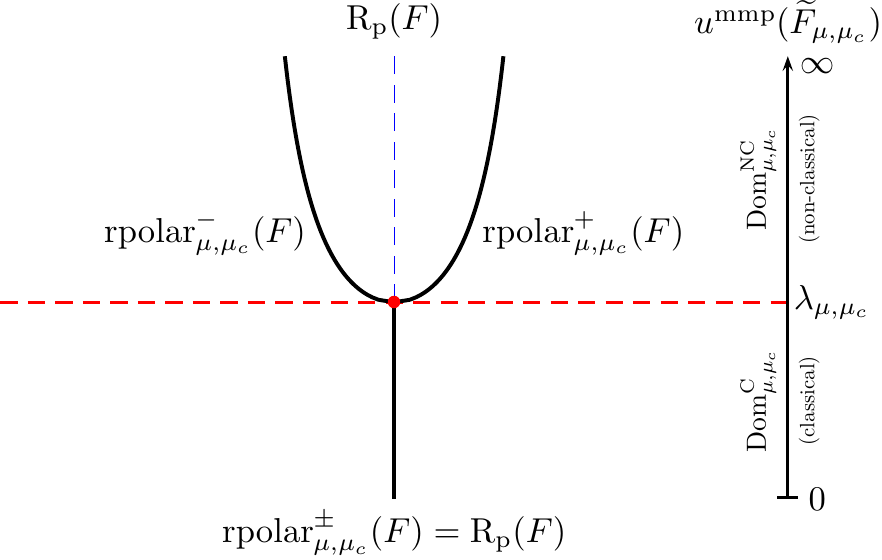}
\end{center}
\caption[Pitchfork bifurcation diagram for $\rpolar_{\mu,\mu_c}^\pm(F)$]{
  \label{fig:branchDiag}  Pitchfork bifurcation diagram for $\rpolar_{\mu,\mu_c}^\pm(F)$
  for $\mu > \mu_c \geq 0$. Let us express the energy-minimizers
  $\rpolar_{\mu,\mu_c}^\pm(F)$ in terms of the maximal mean
  planar stretch $\ump(\widetilde{F}_{\mu,\mu_c})$ of the rescaled
  deformation gradient $\widetilde{F}_{\mu,\mu_c} \eqdef \lambda_{\mu,\mu_c}^{-1}F$.
  For values $F \in \domc_{\mu,\mu_c}$, we have
  $0 < \ump \leq \lambda_{\mu,\mu_c}$ and the polar factor
  $\polar(F)$ is uniquely energy-minimizing. In contrast,
  for $F \in \domn_{\mu,\mu_c}$, $\lambda_{\mu,\mu_c} \leq \ump < \infty$, there
  are two non-classical minimizers $\rpolar^\pm_{\mu,\mu_c}(F)$. In this
  regime, the polar factor is no longer optimal but it is still a
  critical point. At the branching point
  $\ump(\widetilde{F}_{\mu,\mu_c}) = \lambda_{\mu,\mu_c}$
  the minimizers all coincide: $\rpolar^{-}_{\mu,\mu_c}(F) = \polar(F) = \rpolar^{+}_{\mu,\mu_c}(F)$. For $\mu_c \to \mu$, the branching point
  escapes to infinity which asymptotically recovers the behavior
  in the classical parameter range $\mu_c \geq \mu > 0$.}
\end{figure}

The domains of the piecewise definition of $\rpolar_{1,0}^\pm(F)$
in~\coref{cor:rpolar_formula} indicate a certain tension-compression
asymmetry in the material model characterized by the Cosserat
shear--stretch energy $W_{1,0}(R\,;F)$; cf.~\remref{rem:zero_tca}.
We can also make a second important observation. To this end, consider
a smooth curve $F(t):(-\epsilon,\epsilon) \to \GL^+(3)$. If the
eigenvector $q_3(t) \in \S^2$ associated with the smallest singular
value $\sigma_3(t)$ changes its orientation along this curve, then
the rotation axis of $\rpolar_{1,0}^\pm(F)$ flips as well. Effectively,
the sign of the relative rotation angle $\hat{\beta}_{1,0}^\pm(F)$ is
negated which may lead to jumps.  This can happen, e.g., if $F(t)$
passes through a deformation gradient with a non-simple singular
value, but it may also depend on details of the specific algorithm
used for the computation of the eigenbasis.

For the classical range $\mu_c \geq \mu > 0$, the
polar factor and the relaxed polar factor(s) coincide and trivially
share all properties. This is no longer true for the non-classical parameter
range $\mu_c \geq \mu > 0$ and we compare the properties for that range
in our next remark. More precisely, we present a detailed comparison
of the well-known features of the polar factor $\polar$ which are of
fundamental importance in the context of mechanics.

\begin{minipage}{\linewidth}
  \begin{rem}[$\polar(F)$ vs. $\rpolar(F)$ for the non-classical
      \label{rem:polar_vs_rpolar}
    range $\mu > \mu_c \geq 0$] Let $n \geq 2$ and $F \in \GL^+(n)$.
  The polar factor $\polar(F) \in \SO(n)$ obtained from the polar decomposition
  $F = \polar(F)\,U$ is \emph{always unique} and satisfies:
  \begin{equation}
    \begin{matrix*}[l]
      &\text{(Objectivity)}
      &\hspace{.125\linewidth} &\polar(\,Q\cdot F\,)  &= &Q\cdot\polar(F)   &\quad\quad(\forall Q \in \SO(n))\;,\\
      &\text{(Isotropy)}
      &\hspace{.125\linewidth} &\polar(\,F\cdot Q\,)  &= &\polar(F)\cdot Q   &\quad\quad(\forall Q \in \SO(n))\;,\\
      &\text{(Scaling invariance)}
      &\hspace{.125\linewidth} &\polar(\,\lambda\cdot F\,) &= &\polar(F)  &\quad\quad(\forall \lambda > 0)\;,\\
      &\text{(Inversion symmetry)}
      &\hspace{.125\linewidth}                     &\polar(F^{-1}) &= &\polar(F)^{-1}\;. &
    \end{matrix*}
  \end{equation}
  The relaxed polar factor(s) $\rpolar_{\mu,\mu_c}(F) \subset \SO(n)$ is
  \emph{in general multi-valued} and, due to its variational
  characterization, satisfies:
  \begin{equation}
    \begin{matrix*}[l]
      &\text{(Objectivity)}
      &\hspace{.125\linewidth} & \rpolar_{\mu,\mu_c}(\,Q \cdot F\,)    &= &Q \cdot \rpolar_{\mu,\mu_c}(F)    &\quad\quad(\forall Q \in \SO(n))\;,\\
      &\text{(Isotropy)}
      &\hspace{.125\linewidth} &\rpolar_{\mu,\mu_c}(\,F\cdot Q\,)     &= &\rpolar_{\mu,\mu_c}(F)\cdot Q    &\quad\quad(\forall Q \in \SO(n))\;.
    \end{matrix*}
  \end{equation}
  For the particular dimensions $k = 2,3$, our explicit formulae imply
  (cf.~also Part I~\cite{Fischle:2015:OC2D}) that there exist particular
  instances $\lambda^* > 0$ and $F^* \in \GL^+(k)$ for which we have
  \begin{equation}
    \begin{matrix*}[l]
      & \text{(\underline{Broken} scaling invariance)}
      &\hspace{.125\linewidth}
      &\rpolar^\pm_{\mu,\mu_c}(\lambda^* \cdot F^*) &\neq &\rpolar(F^*)
      &, \quad\text{and}\\
      &\text{(\underline{Broken} inversion symmetry)}
      &\hspace{.125\linewidth}
      &\rpolar^\pm_{\mu,\mu_c}({F^*}^{-1})  &\neq &\rpolar(F^*)^{-1} &.
    \end{matrix*}
  \end{equation}
  This can be directly inferred from the partitioning of
  $\GL^+(k) = \domc_{\mu,\mu_c} \,\cup\, \domn_{\mu,\mu_c}$
  and the respective piecewise definition of the relaxed
  polar factor(s), see~\coref{cor:rpolar_formula}.
\end{rem}
We interpret these broken symmetries as a (generalized) tension-compression
asymmetry.
\end{minipage}

\subsection{The reduced Cosserat shear--stretch energy}

We now introduce the notion of a reduced energy which is realized by the
energy-minimizing rotations $\rpolar_{\mu,\mu_c}(F)$; see also~\remref{rem:capriz}.
\begin{defi}[Reduced Cosserat shear--stretch energy]
  The \textbf{reduced Cosserat shear--stretch energy} is defined as
  \begin{equation}
    W_{\mu,\mu_c}^{\rm red}: \GL^+(n) \to \RPosZ,
    \quad\quad
    W_{\mu, \mu_c}^{\rm red}(F) \;\eqdef\; \min_{R\,\in\,\SO(n)} \widehat{W}_{\mu,\mu_c}(R\,;F)\;.
  \end{equation}
\end{defi}
Besides the previous definition, we also have the following equivalent
means for the explicit computation of the reduced energy
\begin{equation}
  \begin{aligned}
    W_{\mu,\mu_c}^{\rm red}(F) &\;=\; \widehat{W}_{\mu,\mu_c}(\rpolar^\pm_{\mu,\mu_c}(F)\,;F)\;,\quad\text{and}\\
    W_{\mu,\mu_c}^{\rm red}(F)
    &\;=\;
    \widehat{W}_{\mu, \mu_c}^{\rm red}(D)
    \;\eqdef
    \min_{\widehat{R} \in \SO(n)} \widehat{W}_{\mu,\mu_c}(\widehat{R}\,; D)
    \;=\;
    \widehat{W}_{\mu,\mu_c}(\widehat{R}_{\mu,\mu_c}^\pm\,;D)\;.
  \end{aligned}
\end{equation}

We now approach the computation of the explicit representation of
$W_{\mu,\mu_c}^{\rm red}(F)$ by means of the equivalent expression
$\widehat{W}^{\rm red}_{\mu,\mu_c}(D)$. For the sake of brevity, we
set $c = \frac{2}{\sigma_1 + \sigma_2} = 1/\ump(F)$ and
$s = \sqrt{1 - c^2}$. This allows us to write the optimal relative
Cosserat rotations in a simple form in the computation of
\begin{equation}
  \widehat{R}_{1,0}^\pm D =
  \begin{pmatrix}
    c & \mp s & 0 \\
    \pm s &c & 0\\
    0 & 0 & 1\\
  \end{pmatrix}
  \begin{pmatrix}
    \sigma_1 & 0 & 0\\
    0 & \sigma_2 & 0\\
    0 & 0 & \sigma_3
  \end{pmatrix}
    = \begin{pmatrix}
    \sigma_1\cdot c & \mp\sigma_2\cdot s & 0\\
    \pm\sigma_1\cdot s &
    \sigma_2\cdot c & 0\\
    0 & 0 & \sigma_3\\
    \end{pmatrix}\;.
\end{equation}

From this, we compute the following symmetric and skew-symmetric parts:
\begin{equation}
  \label{eq:symrd_skewrd}
  \begin{aligned}
    &\sym\left(\widehat{R}_{1,0}^\pm D - \id\right)
    \,\;=\; \begin{pmatrix}
      \sigma_1\cdot c - 1& \frac{d_{12}}{2}\cdot s & 0\\
      \frac{d_{12}}{2}\cdot s& \sigma_2\cdot c - 1  & 0\\
      0 & 0 & \sigma_3 - 1\\
    \end{pmatrix}\;,\quad\text{and}\\
    &\skew\left(\widehat{R}_{1,0}^\pm D - \id\right)
    \;=\;
        {
          \def\arraycolsep{9pt}
          \begin{pmatrix}
            0                         & \mp\frac{s_{12}}{2} \cdot s &\hspace{0.2cm} 0\\
            \pm\frac{s_{12}}{2}\cdot s & 0                        & \hspace{0.2cm}0\\
            0                         & 0                        & \hspace{0.2cm}0\\
      \end{pmatrix}}\;.
  \end{aligned}
\end{equation}

\begin{lem}[The reduced Cosserat shear--stretch energy $\wsymred(F)$ in terms of singular values]
  \label{lem:wred10_singular}
  Let $F \in \GL^+(3)$ and $\sigma_1 > \sigma_2 > \sigma_3 > 0$ the
  ordered singular values of $F$. Then the reduced Cosserat shear--stretch
  energy $\wsymred(F)$ admits the following piecewise representation
  \begin{equation*}
    \wsymred(F) \,=\, \begin{cases}
      \,(\sigma_1 - 1)^2 + (\sigma_2 - 1)^2 + (\sigma_3 - 1)^2 = \hsnorm{U - \id}^2
      &,\,\text{if}\;\,\sigma_1 + \sigma_2 \leq 2,\,\text{i.e.},\, F \in \domc_{1,0}\,,\\
      \,\frac{1}{2}\,(\sigma_1 - \sigma_2)^2 + (\sigma_3 - 1)^2
      &,\,\text{if}\;\, \sigma_1 + \sigma_2 \geq 2,\,\text{i.e.},\,F \in \domn_{1,0}\,.\\
    \end{cases}
  \end{equation*}
\end{lem}

\textit{Proof.}
The classical piece of the energy is easily obtained by inserting the
polar factor $\polar(F)$ into the energy. To compute the non-classical
piece, we first recall that
  \begin{equation*}
    \hsnorm{\sym(\rpolar^{\pm}(F)^TF - \id)}^2 \;=\; \wsymred(F)
    \;=\; \widehat{W}^{\rm red}_{1,0}(D) \;=\; \hsnorm{\sym(\widehat{R}^{\pm}D - \id)}^2\;.
  \end{equation*}
  We compute the expression on the right hand side. To this end, we set $c = \frac{2}{\sigma_1 + \sigma_2}$
  and $s = \sqrt{1 - c^2}$ again and compute the Frobenius matrix norm
  of $\sym(\widehat{R}^{\pm}D - 1)$ which we have derived in~\eqref{eq:symrd_skewrd}. This gives us
  \begin{align*}
      \hsnorm{\sym\left(\widehat{R}^{\pm}D - \id\right)}^2
      &= (\sigma_1 c - 1)^2 + (\sigma_2 c - 1)^2
      + \frac{1}{2}\,(\sigma_1 - \sigma_2)^2 (1 - c^2) + (\sigma_3 - 1)^2\\
      &= \frac{1}{2} \left(4 + (\sigma_1 - \sigma_2)^2 - 4 c (\sigma_1 + \sigma_2) + c^2 (\sigma_1 + \sigma_2)^2\right) + (\sigma_3 - 1)^2\\
      &= \frac{1}{2}(\sigma_1 - \sigma_2)^2 + (\sigma_3 - 1)^2\;.\tag*{$\blacksquare$}
  \end{align*}

Our next step is to reveal the form of the reduced energy for the entire
non-classical parameter range $\mu > \mu_c \geq 0$ which involves the
parameter reduction lemma, but we have to be a bit careful.

\begin{rem}[Reduced energies and the parameter reduction lemma]
  The parameter reduction procedure described
  in~\cite[Lem.~2.2]{Fischle:2015:OC2D} is the key to the minimizers
  for general non-classical material parameters $\mu > \mu_c \geq 0$.
  It might be tempting, but we have to stress that the general form of
  the reduced energy cannot be obtained by rescaling the singular values
  $\sigma_i \mapsto \lambda_{\mu,\mu_c}^{-1}\sigma_i$
  in the singular value representation of $W^{\rm red}_{1,0}$.
\end{rem}

\begin{theo}[$W^{\rm red}_{\mu,\mu_c}$ as a function of the singular values]
  \label{theo:wmm_explicit}
Let $F \in \GL^+(n)$ and $\sigma_1 > \sigma_2 > \sigma_3 > 0$, the ordered
singular values of $F$ and let $\mu > \mu_c \geq 0$, i.e., a non-classical
parameter set. Then the reduced Cosserat shear--stretch energy
$W^{\rm red}_{\mu,\mu_c}: \GL^+(3) \to \RPosZ$ admits the following
explicit representation
\begin{align*}
  W^{\rm red}_{\mu,\mu_c}(F) \,=\,
  \begin{cases}
    \, \mu \left((\sigma_1 - 1)^2 + (\sigma_2 - 1)^2 + (\sigma_3 -
    1)^2\right) = \mu\,\hsnorm{U - \id}^2
    &,\; F \in \domc_{\mu,\mu_c}\;,\\
    \, \frac{\mu}{2}(\sigma_1 - \sigma_2)^2
    + \mu\, (\sigma_3 - 1)^2
    + \frac{\mu_c}{2}\left(\left(\sigma_1 + \sigma_2\right) - \sradmm\right)^2
    - \frac{\mu_c}{2}\cdot\rho_{\mu,\mu_c}^2
    &,\; F \in \domn_{\mu,\mu_c}\;.
  \end{cases}
\end{align*}
\end{theo}

\textit{Proof.} In order to obtain the classical part of the energy it
suffices to insert $\polar{}$ into the energy. For the non-classical
piece, we insert the optimal relative rotations
$\widehat{R}_{\mu,\mu_c}^\pm$
into $\widehat{W}_{\mu,\mu_c}(\widehat{R}\,;D)$. This amounts to replace
$c \mapsto \tilde{c} = \frac{\sradmm}{\sigma_1 + \sigma_2}$ and
$s \mapsto \tilde{s} = \sqrt{1 - \tilde{c}^2}$ in our preparatory
calculation~\eqref{eq:symrd_skewrd}. This yields the following
contributions:
\begin{align}
  \mu\;\hsnorm{\sym(\widehat{R}_{\mu,\mu_c}^\pm D - \id)}^2
  &\;=\; \frac{\mu}{2}\; d_{12}^2
  + \mu\, (\sigma_3 - 1)^2
  + \frac{\mu}{2} (\sradmm -\, 2)^2\;,\\
  \mu_c\;\hsnorm{\skew(\widehat{R}_{\mu,\mu_c}^\pm D - \id)}^2
  &\;=\;
  \frac{\mu_c}{2}\;s_{12}^2\;\tilde{s}^2
  \;=\;
  \frac{\mu_c}{2}s_{12}^2 - \frac{\mu_c}{2} \sradmm^2
  \;.
\end{align}
Finally, adding only the constant part of the symmetric contribution
to the complete contribution due to the skew-symmetric part, we
obtain
\begin{equation}
  \frac{\mu}{2} (\sradmm - 2)^2 +\frac{\mu_c}{2} s_{12}^2 - \frac{\mu_c}{2} \sradmm^2
  =\frac{\mu_c}{2}\left(s^2_{12} - 2\sradmm\right)
  =\frac{\mu_c}{2}\left(s_{12} - \sradmm\right)^2
  - \frac{\mu_c}{2}\;\rho_{\mu,\mu_c}^2\;.
  \tag*{\raisebox{-0.1cm}{$\blacksquare$}}
\end{equation}

The last step of the preceding proof is interesting in its own right.
\begin{rem}[On $\mu_c$ as a penalty weight]
Let us consider the contribution of the skew-term to $W^{\rm red}_{\mu,\mu_c}$
given by
$$
\frac{\mu_c}{2}\left(\left(\sigma_1 + \sigma_2\right) - \sradmm\right)^2
$$
as a penalty term for $F \in \GL^+(3)$ arising for material parameters
in the non-classical parameter range $\mu > \mu_c \geq 0$. This leads to
a simple but interesting observation for strictly positive $\mu_c > 0$.
The minimizers $F \in \GL^+(3)$ \emph{for the penalty term} satisfy
the bifurcation criterion
$$\sigma_1 + \sigma_2 = \sradmm$$
for $\rpolar^\pm_{\mu,\mu_c}(F)$. In this case
$\widehat{R}_{\mu,\mu_c}^\pm = \id$ which implies that
$\widehat{R}_{\mu,\mu_c}^\pm D - \id \in \Sym(3)$, i.e.,
it is symmetric. Hence, the skew-part vanishes entirely which
minimizes the penalty. In numerical applications, a rotation
field $\mrot$ approximating $\rpolar^\pm(F)$ can be expected to
be unstable in the vicinity of the branching point
$\sigma_1 + \sigma_2 \approx \sradmm$.
Hence, a penalty which explicitly rewards an approximation to the
bifurcation point seems to be a delicate property. In strong contrast,
for the case when the Cosserat couple modulus is zero, i.e.,
$\mu_c = 0$, the penalty term vanishes entirely. This hints at a
possibly more favorable qualitative behavior of the model in that
case; cf.~\cite{Neff_ZAMM05}.
\end{rem}

\subsection{Geometric aspects of the reduced Cosserat shear--stretch energy}

We recall that the tangent bundle $T\SO(n)$ is isomorphic to the
product $\SO(n)\times\so(n)$ as a vector bundle. This is commonly referred
to as the left trivialization, see, e.g.,~\cite{Duistermaat:2012:LG}.
With this we can comfortably minimize over the tangent bundle in the following
lemma which sets the course for our next theorem.
\begin{lem}
\label{lem:dist_ctb_technical}
Let $F \in \R^{n\times n}$. Then
\begin{align*}
  \inf_{\substack{R\,\in\,\SO(3)\\ A\,\in\,\so(3)}} \norm{R^TF - \id - A}^2 \quad=\quad \min_{R\,\in\,\SO(3)} \norm{\sym(R^TF - \id)}^2 \quad \defeq \quad \min_{R\,\in\,\SO(3)} W_{1,0}(R\,;F) \;.
\end{align*}
\end{lem}
\begin{proof}
  For all $R\,\in\,\SO(3)$, the infimum of $\hsnorm{\skew(R^TF - \id) - A}^2$
  over all skew symmetric $A$ is obviously attained at $A = \skew(R^TF -
  \id)$. Therefore
  \begin{align*}
    \smash{\inf_{{\substack{R\,\in\,\SO(3)\\A\,\in\,\so(3)}}}} \, \hsnorm{R^TF - \id - A}^2
    &= \inf_{R\,\in\,\SO(3)} \, \inf_{A\,\in\,\so(3)} \Big\lbrace\hsnorm{\sym(R^TF - \id -
      A)}^2 + \hsnorm{\skew(R^TF - \id - A)}^2\Big\rbrace \notag\\
    &= \inf_{R\,\in\,\SO(3)} \, \Big\lbrace \hsnorm{\sym(R^TF - \id)}^2 +
    \inf_{A\,\in\,\so(3)} \hsnorm{\skew(R^TF - \id) - A}^2 \Big\rbrace \notag\\
    &= \inf_{R\,\in\,\SO(3)} \, \hsnorm{\sym(R^TF - \id)}^2\;.
  \end{align*}
  Since $\SO(3)$ is compact and $W_{1,0}(R\,;F)$ is continuous,
  the infimum is attained by a minimizer.\qedhere
\end{proof}

The preceding lemma leads us to a nice geometric characterization of
the reduced Cosserat shear--stretch energy which we find quite remarkable.
It might even be useful for the case $n \geq 4$ although this is somewhat
far-fetched.
\begin{cor}[Characterization of $W_{1,0}^{\rm red}$ as a distance]
  \label{cor:wred10_distance}
  Let $n \geq 2$ and consider $F \in \GL^+(n)$ with singular values
  $\sigma_1 \geq \sigma_2 \geq \ldots \geq \sigma_n > 0$ not necessarily
  distinct. Then the reduced Cosserat shear--stretch energy
  $W_{1,0}^{\rm red}: \GL^+(n) \to \RPosZ$
  admits the following characterization as a distance:
  \begin{equation}
    W^{\rm red}(F) \quad=\quad \dist_{\rm euclid}^2\big(F,\, \SO(n)\left(\id + \so(n)\right)\big)\;.
  \end{equation}
  Here, $\dist_{\rm euclid}$ denotes the euclidean distance function.
\end{cor}
\textit{Proof.} First note that
\begin{align*}
  W^{\rm red}_{1,0}(F) &\eqdef \min_{R\,\in\,\SO(n)}\norm{\sym(R^TF - \id)}^2
  \stackrel{(\text{Lem. \ref{lem:dist_ctb_technical}})}{=} \smash{\inf_{\substack{R\,\in\,\SO(3)\\ A\,\in\,\so(3)}}} \norm{R^TF - \id - A}^2\\
  &= \inf_{\substack{R\,\in\,\SO(3)\\ A\,\in\,\so(3)}} \norm{R(R^TF - \id - A)}^2\;.
\end{align*}
The last step ist justified by the orthogonal invariance of the Frobenius
norm $\norm{\cdot}$. Carrying out the multiplications on the right hand
side, we are lead to the conclusion
\begin{align*}
  \inf_{\substack{R\,\in\,\SO(3)\\ A\,\in\,\so(3)}} \norm{F - R\,\id - RA}^2
  = \inf_{\substack{R\,\in\,\SO(3)\\ A\,\in\,\so(3)}} \norm{F - R(\id + A)}^2
  \defeq \dist_{\rm euclid}^2(F, \SO(n)\left(\id + \so(n)\right))\;.\tag*{\raisebox{-0.4cm}{$\blacksquare$}}
\end{align*}

\begin{figure}[t]
  \begin{center}
    \includegraphics[width=7.2cm,clip=false,trim=3.2cm 2cm 3.2cm 3.2cm]{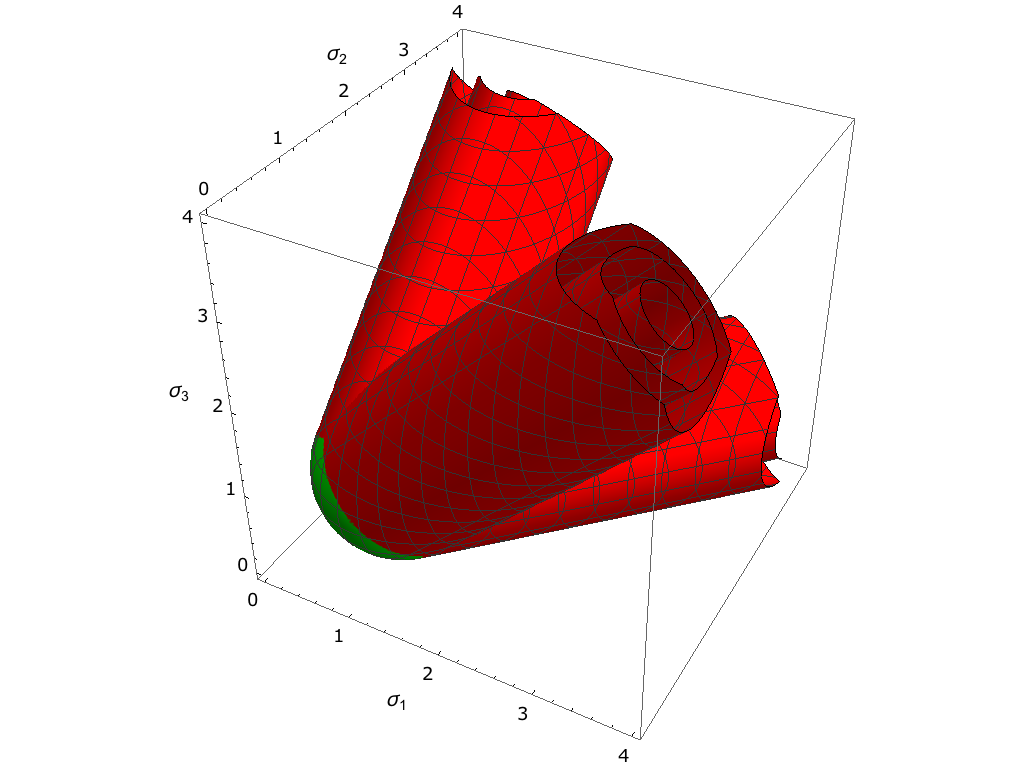}
    \includegraphics[width=7.2cm,clip=false,trim=3.2cm 2cm 3.2cm 3.2cm]{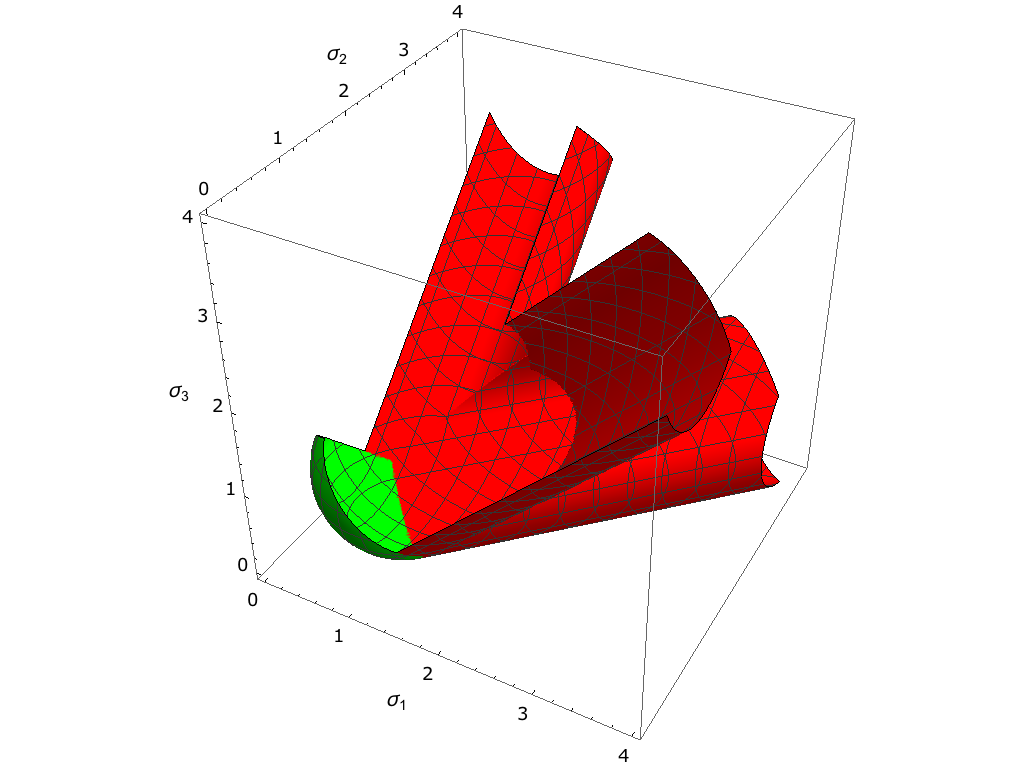}
  \end{center}
  \caption[Energy isosurfaces of {$W^{\rm red}_{1,0}$} in the space of
    singular values of $F$]{\label{fig:wred_isosurf}    Energy isosurfaces of $W^{\rm red}_{1,0}$ considered as
    a function of the \emph{unordered} singular values
    $\sigma_1,\sigma_2,\sigma_3 > 0$ of $F \in \GL^+(3)$.
    The displayed contour levels are $0.1$, $0.4$ and $0.8$.
    On the right, we have removed a piece from the non-classical
    cylindrical parts (red) of the energy level $0.8$ which reveals
    the spherical shell of the classical part (green). Note that
    a computation of these level surfaces via Monte Carlo minimization
    yields the same result (but at a much lower resolution).}
\end{figure}

\begin{rem}[Zero reduced energy and tension-compression asymmetry]
  \label{rem:zero_tca}
  A sharp look at~\leref{lem:wred10_singular} is sufficient to see that
  the $0$-energy level of $W^{\rm red}_{1,0}$ precisely corresponds
  to singular value tuples of the form
  $(s,s,1)$, $s \in [1,\infty)$.\footnote{Technically, our derivation
      of~\leref{lem:wred10_singular} does not extend to the case of
      multiple singular values, but the characterization as
      a distance function in~\coref{cor:wred10_distance} does not
      have this limitation.}
    In our~\figref{fig:wred_isosurf} tuples of this type (and permutations
    thereof) correspond to the axes of the cylindrical sheets of the
    isosurfaces. Let us now consider $X = R(\id + A)$, $R \in \SO(3)$,
    $A \in \so(3)$, which has the squared singular values
    $(\sigma_1^2,\sigma_2^2,\sigma_3^2) = \left(1 + \hsnorm{A}^2, 1 + \hsnorm{A}^2, 1\right)$. Clearly, such a matrix $X$ does not generate
    any reduced Cosserat shear--stretch energy at all -- in perfect accord
    with~\coref{cor:wred10_distance}. Geometrically, $U(X) \eqdef \sqrt{X^TX}$
    induces a homogeneous blow-up (i.e., a rescaling of arbitrary positive
    magnitude) of the plane of maximal strain $\mathrm{P}^{\rm mp}(X)$
    while preserving the distance of any given point to this plane.
    Furthermore, there is no possibility of similar energy savings in the
    compressive range for $F \in \GL^+(3)$ where the classical piece
    of $W^{\rm red}_{1,0}$ is active. It seems to us that this makes
    a good case for a quite remarkable type of tension-compression asymmetry.
\end{rem}

\subsection{Alternative criteria for the existence of non-classical solutions}
For $\mu > \mu_c > 0$, i.e., for strictly positive $\mu_c > 0$,
the singular radius satisfies $\sradmm \eqdef \frac{2\mu}{\mu - \mu_c} > 2$.
We now define a quite similar constant, namely
\begin{align}
\zeta_{\mu,\mu_c} \eqdef \sradmm - \;\rho_{1,0}
= \frac{2\mu_c}{\mu - \mu_c} > 0\;.
\end{align}
Furthermore, we define the $\epsilon$-neighborhood of a set
$\mathcal{X} \subseteq \R^{n\times n}$ relative to the euclidean
distance function as
$$
N_{\epsilon}(\mathcal{X})
\;\eqdef\;
\setdef{Y \in \R^{n \times n}}{\dist_{\rm euclid}(Y, \mathcal{X}) < \epsilon}\;.
$$

\begin{figure}[h!]
  \begin{center}
    \scalebox{0.75}{
    \begin{tikzpicture}
      \node (Pic) at (0,0)
            {\includegraphics[width=8cm]{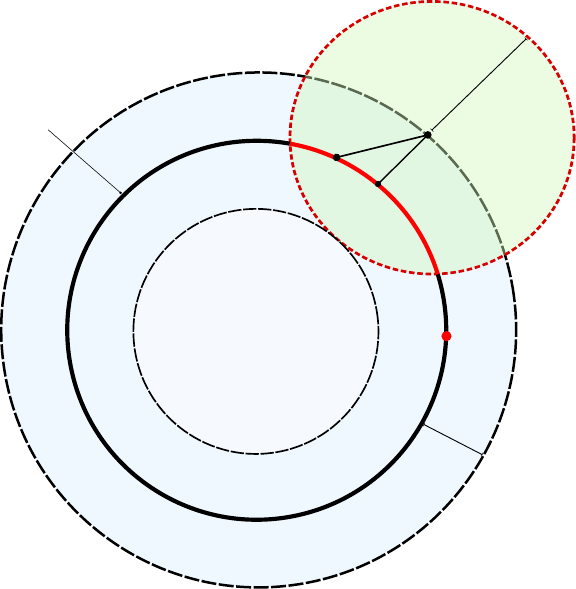}};
            \node[scale=1.2] (Id) at (2.5,-0.55) {$\id$};
            \node[rotate=-33] (UEpsSO3) at (-2.1, -3.05) {$N_\epsilon(\SO(3))$};
            \node[rotate=-25] (HalfEps) at (2.35, -1.85) {$\epsilon$};
            \node[rotate=45] (Eps) at (2.6, 3.1) {$\delta$};
            \node (SO3) at (-3.4,2.6) {$\SO(3)$};
            \node (F) at (2.0, 2.6) {$F$};
            \node (UEpsF) at (3.25,2) {$N_{\delta}(F)$};
            \node (R) at (0.5,1.7) {$R$};
            \node[rotate=12] (HR) at (1.2,2.2) {$F-R$};
    \end{tikzpicture}}
  \end{center}
  \caption{Illustration of a euclidean $\epsilon$-neighborhood of $\SO(3) \subset
    \R^{3\times 3}$.}
\end{figure}

\begin{lem}[Classical $\SO(3)$-neighborhood for $\mu_c > 0$]
  Let $\mu > \mu_c > 0$, $F \in \GL^+(3)$ and
  $\zeta_{\mu,\mu_c} \eqdef \frac{2\mu_c}{\mu - \mu_c} > 0$.
  Then we have the following inclusion
\begin{align}
  N_{\frac{1}{2}\zeta_{\mu,\mu_c}^2}(\SO(3)) \quad\subset\quad \domc_{\mu,\mu_c}\;.
\end{align}
In other words, for all $F \in \GL^+(3)$ satisfying $\dist_{\rm euclid}(F, \SO(3)) = \hsnorm{U - \id}^2 < \frac{1}{2}\zeta_{\mu,\mu_c}^2$,
the polar factor $\polar{}$ is the unique minimizer of $W_{\mu,\mu_c}(R\,;F)$.
\begin{proof}
Since $\dist_{\rm euclid}^2(F,\SO(3)) = \norm{U - \id}^2 =\sum_{i=1}^3 (\sigma_i - 1)^2$
by Grioli's theorem~\cite{Neff_Grioli14}, we find
\begin{align*}
  \dist_{\rm euclid}^2(F,\SO(3)) < \frac{1}{2}\;\zeta_{\mu,\mu_c}^2
  \quad\Longrightarrow&\quad
  2\left((\sigma_1 - 1)^2 + (\sigma_2 - 1)^2 + (\sigma_3 - 1)^2\right) < \zeta_{\mu,\mu_c}^2\\
  \quad\Longrightarrow&\quad
  2\left((\sigma_1 - 1)^2 + (\sigma_2 - 1)^2\right) < \zeta_{\mu,\mu_c}^2\;.
\end{align*}
Further, $0 \leq (a - b)^2 = a^2 + b^2 - 2ab$ implies $2\,(a^2 + b^2) \geq a^2 + b^2 + 2ab$
and it follows that
\begin{align}
(\sigma_1 - 1)^2 + (\sigma_2 - 1)^2 + 2(\sigma_1 - 1)(\sigma_2 - 1) < \zeta_{\mu,\mu_c}^2\;.
\end{align}
Completing the square and taking square roots on both sides, we find
\begin{align}
\left((\sigma_1 - 1) + (\sigma_2 - 1)\right)^2 < \zeta^2_{\mu,\mu_c}
\quad\Longrightarrow\quad
\pm\abs{(\sigma_1 - 1) + (\sigma_2 - 1)} < \zeta_{\mu,\mu_c}\;.
\end{align}
Inserting $\zeta_{\mu,\mu_c} \eqdef \sradmm -\;2$, we obtain
$(\sigma_1 - 1) + (\sigma_2 - 1) \;<\; \rho_{\mu_,\mu_c} -\; 2$.
This implies $\sigma_1 + \sigma_2 < \rho_{\mu_,\mu_c}$
and hence $F \in \domc_{\mu,\mu_c}$.\qedhere
\end{proof}
\end{lem}
Note that the preceding proof can be quite easily adapted to the planar
case $n = 2$ presented in~\cite{Fischle:2015:OC2D}.

\begin{lem}
\label{lem:disc:SL3_dom_nc}
Let $F \in \SL(3)$, i.e., $\det{F} = \sigma_1\sigma_2\sigma_3 = 1$,
where $\sigma_1 \geq \sigma_2 \geq \sigma_3 > 0$ are ordered
singular values of $F$, not necessarily distinct. Then
\begin{align}
  \SL(3) \quad\subset\quad \domn_{1,0}\;,
\end{align}
i.e., $F$ induces a strictly non-classical minimizer. Equivalently,
$\det{F} = 1$ implies the estimate $\sigma_1 + \sigma_2 \geq 2$.
\begin{proof}
The inequality for the geometric and arithmetic mean shows that
\begin{equation}
  \frac{\sigma_1 + \sigma_2 + \sigma_3}{3}
  \quad\geq\quad
  (\sigma_1\,\sigma_2\,\sigma_3)^\frac{1}{3}
  \quad=\quad
  1
\;.
\end{equation}
It follows that $\sigma_1 + \sigma_2 \geq 3 - \sigma_3$ which implies the claim
for $\sigma_3 \leq 1$. Due to the ordering $\sigma_1 \geq \sigma_2 \geq \sigma_3 > 0$, the case $\sigma_3 > 1$ contradicts our assumption $\det{F} > 1$.\qedhere
\end{proof}
\end{lem}

\begin{rem}
  \label{rem:disc:SL3_dom_nc_strict}
  If we make the stronger assumption $\sigma_1 > \sigma_2 > \sigma_3 > 0$,
  we obtain a strict inequality $\sigma_1 + \sigma_2 > 2$. In that case,
  $F \in \domn_{1,0} \setminus \domc_{1,0}$ is strictly non-classical.
\end{rem}

\begin{cor}
  \label{cor:SL3_dom_nc_strict}
  Let $\mu > 0$,
    $F \in \SL^+(3)$ and assume
  that $\sigma_1 > \sigma_2 > \sigma_3 > 0$. Then
  \begin{equation}
    F \quad\in\quad \domn_{\mu,0}\setminus \domc_{\mu,0}\;,
  \end{equation}
  i.e., the minimizers $\rpolar_{\mu,0}^\pm(F) \neq \polar{}$ are \emph{strictly}
  non-classical.
  \begin{proof}
    Since $\lambda_{\mu,0} = 1$, it follows that $\widetilde{F}_{\mu,0} = F$.
    Further $\rho_{\mu,0} = \rho_{1,0}$. Thus, we are in the hypotheses of the
    preceding~\leref{lem:disc:SL3_dom_nc} for the case where the inequality is
    strict, see~\remref{rem:disc:SL3_dom_nc_strict}.
  \end{proof}
\end{cor}

\subsection{Application}
\label{sec:discussion:application}
Let us now give a short application to our previous findings. We consider a
so-called volumetric-isochoric split for the geometrically nonlinear Cosserat
shear--stretch energy. Note that this material model appears in a variety of
contexts, see, e.g.,~\cite{Eremeyev:2012:FMM,Boehmer:2015:SS,Neff:2015:EGNC,Lankeit:2015:IC,Neff_Muench_transverse_cosserat08,Neff_Muench_magnetic08,Neff_Cosserat_plasticity05,Forest:1997:CSC,Sansour:1998:TEVC,Sansour:2008:NCC}, and, recently~\cite{Skatulla:2013:CCMS,Matteo:2015:MMD,Blesgen:2013:DPF,Blesgen:2014:DPC}. Further, similar expressions for the strain energy have been considered in
the context of plate and shell theories, see, e.g.,~\cite{Pietraszkiewicz:2009:VPN,Pietraszkiewicz04,Pietraszkiewicz08,Neff_plate04_cmt,Birsan_Neff_MMS_2013,Birsan_Neff_JElast2013,Sander:2014:NGNCS}.

Let us introduce the isochoric projection $F \mapsto F_{\rm iso}
\eqdef \frac{F}{\det{F}^{1/3}} \in \SL(3)$ of the deformation gradient
$F \in \GL^+(3)$ which can also be applied to $\mstretch \eqdef
R^TF$. With this notation, we obtain
\begin{align*}
  W(\mstretch)
  &= \underbrace{\mu\, \hsnorm{\sym\left(\frac{\mstretch}{\det{\mstretch}^{1/3}} - \id\right)}^2
    \;+\;
    \mu_c\,\hsnorm{\skew\left(\frac{\mstretch}{\det{\mstretch}^{1/3}} - \id\right)}^2}_{\text{``Cosserat shear--stretch energy''}}
    \;+\;
    \underbrace{\vphantom{\hsnorm{\sym\left(\frac{\mstretch}{\det{\mstretch}^{1/3}} - \id\right)}^2} h(\det{\mstretch})}_{\text{``volumetric contribution''}}\\[0.2cm]
  &= \mu\, \hsnorm{\sym\left(\mrot^T\frac{F}{\det{F}^{1/3}} - \id\right)}^2
    \;+\;
    \mu_c\,\hsnorm{\skew\left(\mrot^T\frac{F}{\det{F}^{1/3}} - \id\right)}^2
    \;+\;
    h(\det{R^TF})\\[0.2cm]
 &= \mu\, \hsnorm{\sym\left(\mrot^TF_{\rm iso} - \id\right)}^2
    \;+\;
    \mu_c\,\hsnorm{\skew\left(\mrot^TF_{\rm iso} - \id\right)}^2
    \;+\;
    h(\det{F})\;.
\end{align*}

The results of the previous subsections, allow us to determine the
optimal Cosserat rotations for the split energy
$W(\mstretch) \;=\; W(\mrot^TF) \;=\; \wmm(R\,;F_{\rm iso}) \;+\; h(\det{F})$.
Note first that the additional volumetric contribution $h(\det{F})$
penalizes volume change by a scalar function $h: \RPos \to \RPosZ$
which is constant with respect to $\mrot \in \SO(3)$. Therefore, this
formulation still gives rise to the same optimal Cosserat rotations
$$\rpolar_{\mu,\mu_c}(F_{\rm iso})
\eqdef \argmin{R\,\in\,\SO(3)}{\wmm(R\,;F_{\rm iso})}
= \argmin{R\,\in\,\SO(3)}{\Big\{\wmm(R\,;F_{\rm iso})\;+\; h(\det{F})\Big\}}\;.$$

We can now make an interesting observation. To this end, let $\epsilon > 0$
and consider diagonal matrices of the type
\begin{equation*}
  D_\epsilon \;\eqdef\;
  \begin{pmatrix}
    \sradmm -\, 1 + \epsilon & 0 & 0\\
    0 & 1 & 0\\
    0 & 0 & (\sradmm -\, 1 + \epsilon)^{-1}
  \end{pmatrix}
  \;\in\; \SL(3)\;.
\end{equation*}
The required ordering $\sigma_1 > \sigma_2 > \sigma_3 > 0$ follows from
$\sradmm \eqdef \frac{2\mu}{\mu - \mu_c} \geq \frac{2\mu}{\mu} = 2$ and holds
for the entire non-classical parameter range $\mu > \mu_c \geq 0$.
Obviously, we have
\begin{equation*}
  \sigma_1 + \sigma_2 = \sradmm + \;\epsilon > \sradmm
  \quad\Longrightarrow\quad
  D_\epsilon \in \domn_{\mu,\mu_c} \setminus \domc_{\mu,\mu_c}\;.
\end{equation*}
Hence, the intersection
$\SL(3) \cap \left(\domn_{\mu,\mu_c} \setminus \domc_{\mu,\mu_c}\right) \neq \emptyset$
is never empty since it contains $D_\epsilon$ for all $\epsilon > 0$.
Furthermore, the associated optimal Cosserat rotations are \emph{strictly}
non-classical, i.e., $\rpolar^\pm_{\mu,\mu_c}(D_\epsilon) \neq \polar(D_\epsilon) = \id$.

Hence, in order to assure that there can be no strictly non-classical optimal
Cosserat rotations (for whatever reason) one has to consider
material parameters from the classical parameter range $\mu_c \geq \mu > 0$.
In this case the Cosserat couple modulus $\mu_c$ dominates the Lam\'e shear
modulus $\mu$ and Grioli's theorem assures that the polar factor $\polar(F)$
is always uniquely optimal~\cite{Neff_Grioli14}.

In the distinguished limit case $\mu_c = 0$, the volumetric-isochoric
split precludes the previously observed tension-compression
asymmetry. In this particular scenario, \coref{cor:SL3_dom_nc_strict}
shows that the optimal rotations are \emph{always}
non-classical. Since no bifurcation of the optimal rotations occurs,
there can be no qualitatively different energetic response under
tension and compression for $\mu_c = 0$; cf. also~\cite{Neff_ZAMM05}
for a discussion of other implications of a zero Cosserat couple
modulus.

Last but not least, we want to mention that our proposed explicit formulae
for optimal Cosserat rotations may also lead to improved stability and performance
in full scale 3D nonlinear finite element computations for media with
rotational microstructure. We expect them to be especially useful for
the highly interesting and numerically challenging case of a material
with small internal length scale $L_{\rm c} > 0$. If, in addition,
the volumetric contribution is independent of the rotation (see above),
then the optimal Cosserat rotations $\rpolar^\pm_{\mu,\mu_c}(F)$
proposed in~\coref{cor:rpolar_formula} can be expected to be
ideal candidates for the initialization of the Newton--iterations
for the field of microrotations $\mrot: \Omega \subset \R^n \to \SO(n)$,
$n = 2,3$.

 \countres
\section{Dissection of critical point structure and computational validation
  of optimality}
\label{sec:validation}

\begin{figure}[tb]
  \begin{center}
    \includegraphics[width=3cm]{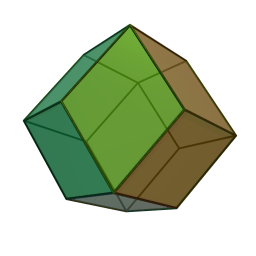}
    \includegraphics[width=3cm]{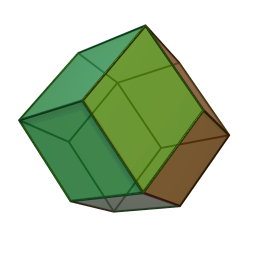}
    \includegraphics[width=3cm]{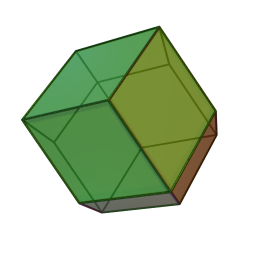}
    \includegraphics[width=3cm]{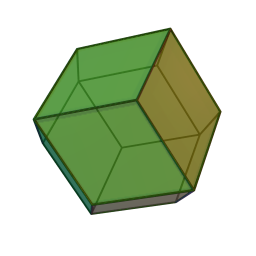}\\
  \end{center}
  \caption[A rhombic dodecahedron in the space of singular values]{\label{fig:dodecahedron}
    A rhombic dodecahedron placed in the space of unordered singular
    values $(\sigma_1,\sigma_2,\sigma_3) \in \R^3$
    of $F \in \GL^+(3)$ gives rise to a beautiful geometric
    characterization of the classical and non-classical
    domains $\domc_{1,0}$ and $\domn_{1,0}$. Pick a face
    and displace it in normal direction while scaling it
    by its distance to the origin. This creates a convex cone
    with the scaled faces as cross-sections which intersects
    the polytope. The part of the cone inside the polytope
    corresponds to singular values $(\sigma_1,\sigma_2,\sigma_3) \in \R^3$
    in $\domc_{1,0}$. The exterior part corresponds to $\domn_{1,0}$.
    On the picked face itself $\domc_{1,0} \cap \domn_{1,0}$, the two
    branches coincide. (This is how we discovered the non-classical
    bifurcation behavior of $\rpolar^\pm_{1,0}(F)$ in dimension $n\!=\!3$.)
  }
  \label{fig:rhombic_dodecahedron}
\end{figure}

We recall that our primary objective for the present work is to
derive a formula (or algorithm) which allows to compute the set of
optimal Cosserat rotations $\rpolar_{\mu,\mu_c}(F) \subset \SO(3)$, i.e., the
rotations which minimize the Cosserat shear--stretch energy $\wmm(R\,;F)$
for given $F \in \GL^+(3)$ and weights $(\mu,\mu_c)$ in the non-classical
parameter range $\mu > \mu_c \geq 0$. In the first two sections of this
contribution, we have hopefully convinced our avid reader that it suffices
to solve~\probref{prob:relative_rhat} in order to determine the optimal
Cosserat rotations which then solve our original~\probref{intro:prob_wmm}.
However, in order to simultaneously cross-validate our theoretical
derivation (this includes the parameter reduction presented in
Part I~\cite[Lem.~2.2, p.~4]{Fischle:2015:OC2D}), we have
based our final validation on~\probref{intro:prob_wmm}. This bypasses
all simplification steps which we have used in order to derive
the formula for $\rpolar_{\mu,\mu_c}^\pm(F)$ proposed
in~\coref{cor:rpolar_formula}, but is costly due to the large
parameter space.

\subsection{Interactive analysis of the critical point structure}
The solution of the Euler--Lagrange equations~\eqref{eq:EL_quat}
with the computer algebra package~\mathematica~returns the $32$
critical points compiled in Appendix~\ref{sec:appendix}. Note that
\mathematica~automatically verifies that the obtained symbolic
expressions are indeed solutions for~\probref{prob:lagrange}.

These symbolic solutions give rise to $32$ critical branches
$\hat{q}^{(i)}: \mathrm{Dom}(\hat{q}^{(i)}) \to \S^3$,  $1 \leq i \leq 32$.
Note that, we can discard $16$ of the branches right away
since they are redundant. This is due to the antipodal
identification of quaternions under the covering map
$\pi: \S^3 \to \SO(3)$. The critical branches are
associated with the \emph{lifted} Cosserat shear--stretch energy
formulation $\widehat{W}^\sharp_{\mu,\mu_c}(\hat{q}\,;D)$.
In particular, they project to \emph{relative} rotations
parametrized by a diagonal matrix $D = \diag(\sigma_1,\sigma_2,\sigma_3)$.
In what follows, we identify the space of \emph{unordered} singular
values $(\sigma_1, \sigma_2, \sigma_3)$ of $F \in \GL^+(3)$ with $\R^3$.
Further, we take the liberty to identify diagonal
matrices $D = \diag(\sigma_1,\sigma_2,\sigma_3)$ with points in
$\R^3$ and shall even write $D \in \R^3$.

This allows us to write
$\mathrm{Dom}(\hat{q}^{(i)}) \subset \R^3$, $i = 1,\ldots,16$,
for the maximal domain of definition of the $i$-th
critical branch.\footnote{Technically, these maximal domains are implicitly
defined by the requirement that the critical coefficient tuple computed
by \mathematica~is real-valued at the point $D \in \R^3$, i.e.,
$(w^{(i)},x^{(i)},y^{(i)},z^{(i)}) \in \R^4$.} If the solution set
is complete (cf. Appendix~\ref{sec:appendix} for a discussion),
then, up to a set of measure zero, we must have\footnote{Equivalently, for a given $F \in \GL^+(3)$ with distinct
  singular values, at least one of the critical branches
$\hat{q}^{(i)}: \mathrm{Dom}(\hat{q}^{(i)}) \to \S^3$
must be energy-minimizing since the Cosserat shear--stretch
energy attains its minimum on $\SO(3)$.}
$$\bigcup_{1\leq i \leq 16} \mathrm{Dom}(\hat{q}^{(i)}) = \R^3\;.$$

Initially, still stumbling in the dark, we compared the critical
branches by comparing the different realized energy levels
given by $\widehat{W}_{1,0}(\hat{q}^{(i)}\,; D)$, $i = 1,\ldots,16$,
for random tuples $(\sigma_1, \sigma_2, \sigma_3) \in \R^3$.
This allowed us to construct a three-dimensional map for the
space of unordered singular values by associating the index
set of the energy-minimizing critical branches at
$D = (\sigma_1, \sigma_2, \sigma_3)$ with the point
$D \in \R^3$ in the parameter space. We then mapped each
of these index sets to a unique color and subsequently explored
the parameter space visually. This three-dimensional
``optimal branch map'' allowed us to isolate the energy-minimizing
critical branches corresponding to $\rpolar^\pm_{1,0}(F)$;
cf.~\figref{fig:dodecahedron} which describes the structure that
appeared.

Furthermore, we compared the following minimal energy levels
\begin{equation}
  \min_{1 \leq i \leq 16} \widehat{W}^\sharp_{1,0}(\hat{q}^{(i)}(D)\,; D)\;,
  \quad\text{and}\quad
  \min_{\hat{q}\,\in\S^3} \widehat{W}^\sharp_{1,0}(\hat{q}\,; D)\;,
\end{equation}
using a Monte Carlo approximation of the right hand side
(which we describe in detail in the next subsection). This
allowed us to detect discrepancies which can only arise
due to an incomplete set of critical branches. Note that
during our whole investigation, we never encountered any
such discrepancy. This is a strong indication that
the set of critical points computed by \mathematica~is
in fact complete, as one would expect.

Our next step is to turn our previously described approach into a
more systematic computational validation of the optimality of
the proposed candidates $\rpolar^\pm_{\mu,\mu_c}(F)$.

\subsection{Validation of optimality by Monte Carlo statistical sampling}
We now describe a serious computational approach for the
validation of the optimality of our proposed candidate
formula $\rpolar^\pm_{\mu,\mu_c}(F)$. This approach relies on a
well-known, rather simplistic, but highly useful (in low
dimensions) method for the generation of uniformly distributed
random rotations due to~\cite{Shoemake:1992:URR}.

In what follows, we let $K \eqdef [-1,1]^4 \subset \R^4$ denote
a hypercube of sidelength $2$ centered about the origin and
define $\mathbb{B}^4 \eqdef \setdef{x \in K}{\norm{x} \leq 1}$, i.e.,
as the closed unit ball in $\R^4$. Further, we let $X_K$ denote a
uniformly distributed random variable with values in $K$
and introduce $X_{\mathbb{B}^4}$ as the restriction of $X_K$ to
the unit ball. Then, $X_{\mathbb{B}^4} \eqdef \restrict{X_K}{\mathbb{B}^4}$
is uniformly distributed. The restriction can be defined by
rejection sampling, i.e., we reject all realizations in
$K \setminus \mathbb{B}^4$ which lie outside of the ball,
but accept the first realization inside of
$\mathbb{B}^4 \subset K$; see~\figref{fig:rejection_sampling}
for an example in the plane.

\begin{theo}[Rejection sampling for $\S^3$]
  The random variable $X_{\S^3} \eqdef \frac{X_{\mathbb{B}^4}}{\norm{X_{\mathbb{B}^4}}}$
  obtained by normalization is uniformly distributed on $\S^3$ with respect
  to the Lebesgue measure on the sphere which we denote by $\dV_{\S^3}$.
\end{theo}
\begin{proof}
  This is a standard method which performs quite well in low space
  dimensions, see, e.g.,~\cite{Muller:1959:GPS,Hicks:1959:GPS}
  or~\cite{Marsaglia:1972:CPS} and references therein.
\end{proof}

\begin{figure}[t]
  \begin{center}
    \includegraphics[width=3.9cm]{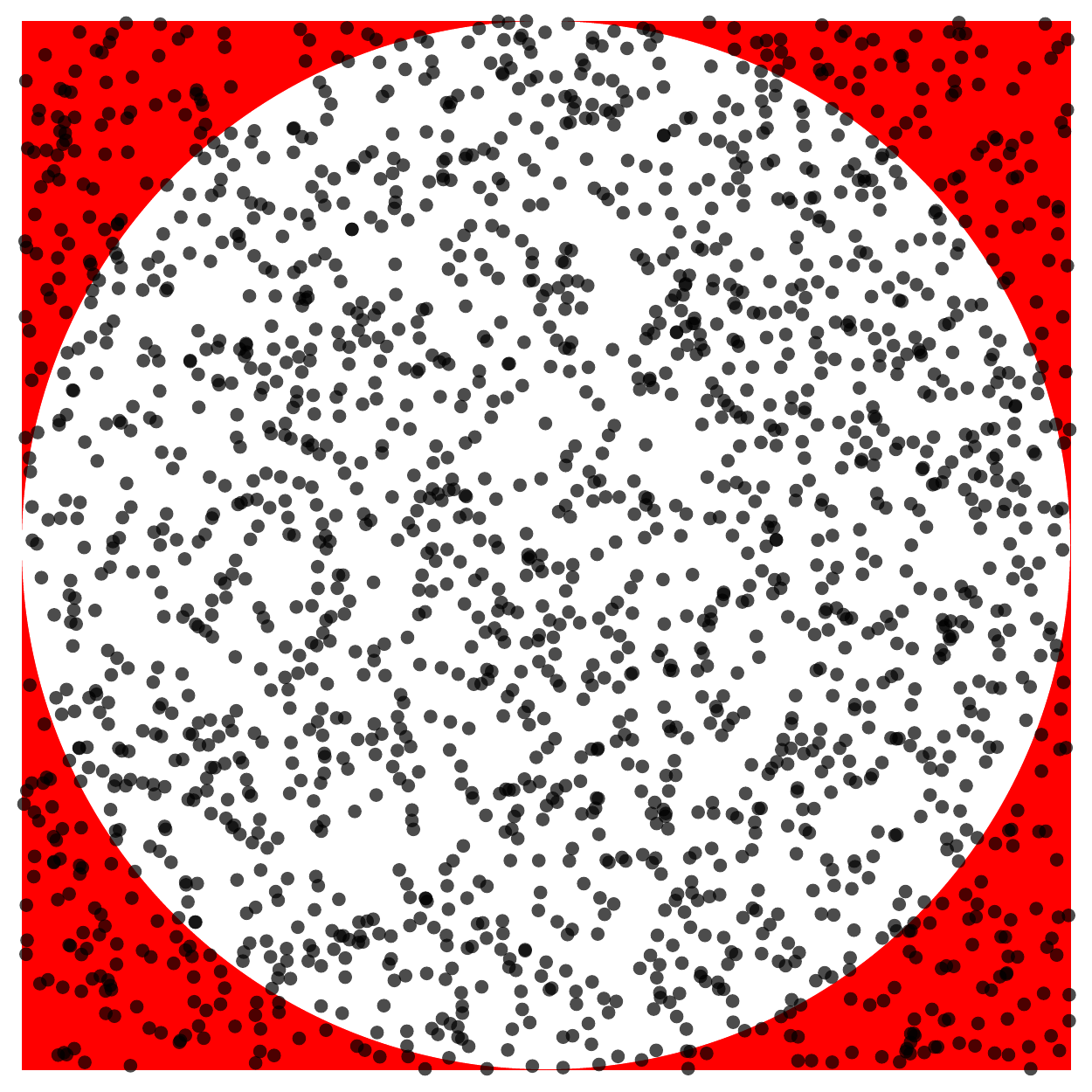}
    \qquad
    \includegraphics[width=3.9cm]{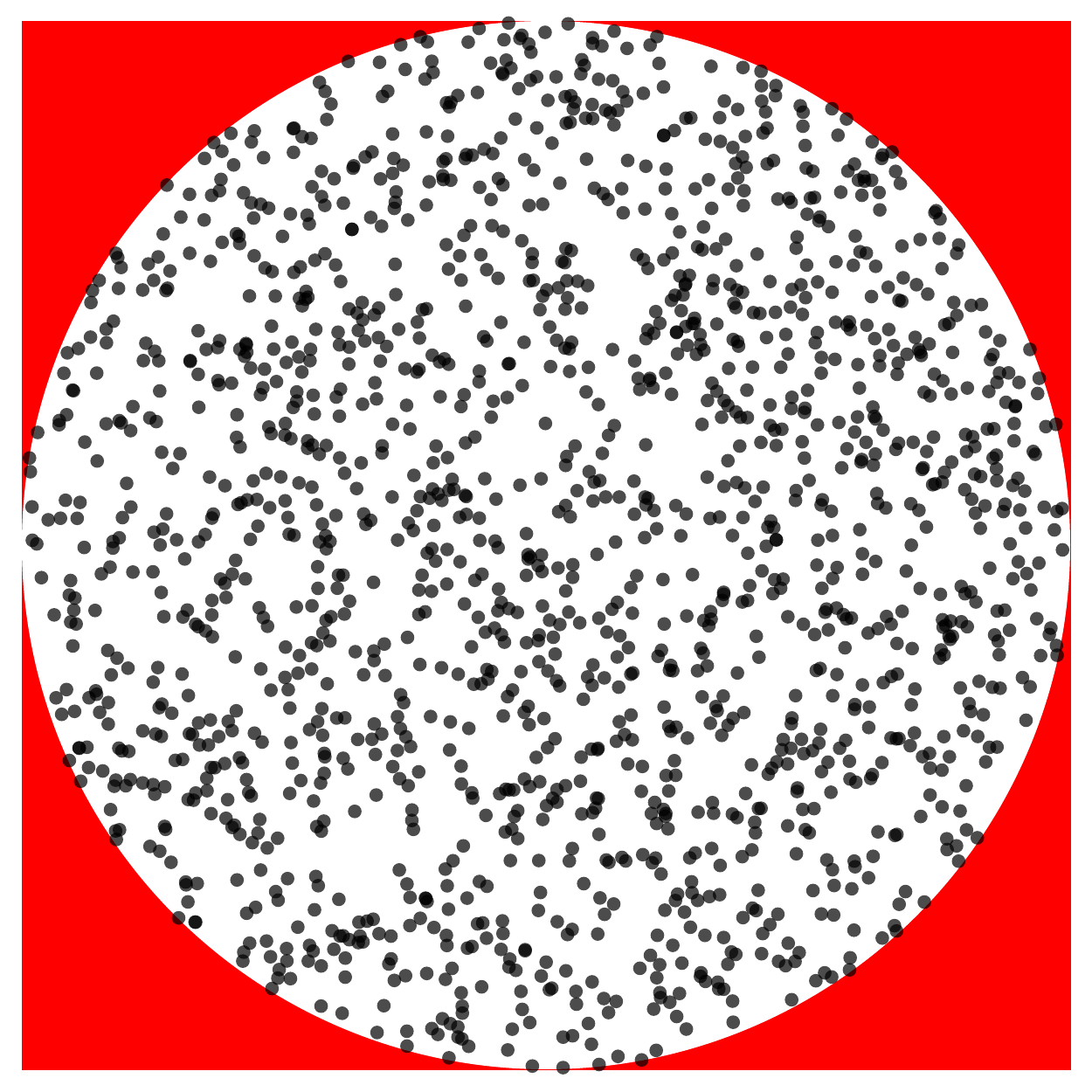}
    \qquad
    \includegraphics[width=3.9cm]{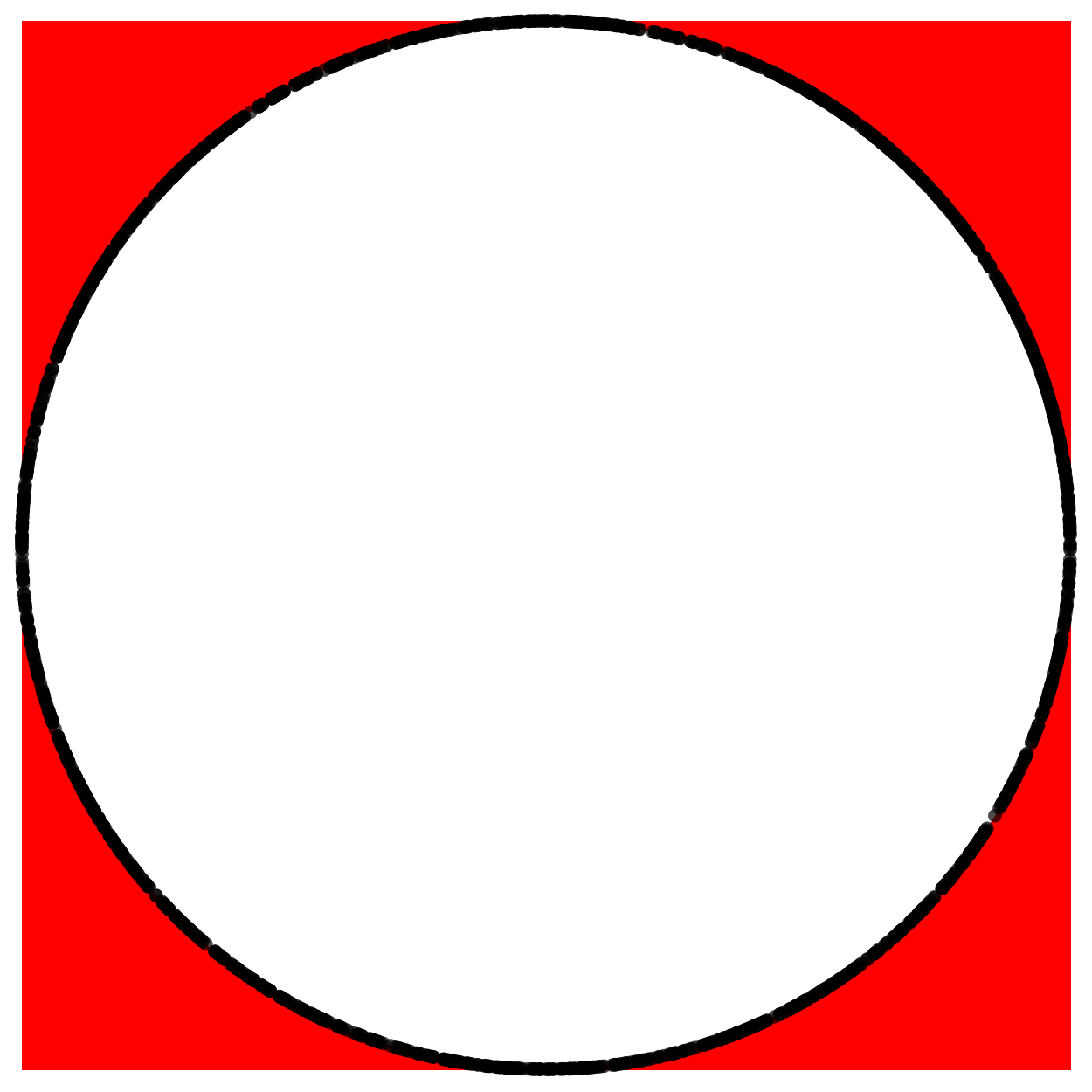}
  \end{center}
  \caption[Uniform sampling of $\S^1$ via rejection sampling]{
    \label{fig:rejection_sampling}    Uniform sampling of the circle $\S^1$ by rejection sampling.
    First, the unit square $[-1,1]^2$ is uniformly
    sampled (here, with $2000$ samples). Then all samples $p$ with $\norm{p} > 1$
    in the red domain are rejected. Finally, the remaining samples in the unit
    disk are normalized, i.e., $p \mapsto \frac{p}{\norm{p}} \in \S^1$.
    This yields a uniform distribution on the boundary circle $\S^1$.
    Although this approach can be generalized to higher-dimensional
    spheres its performance does not scale to high dimensions.}
\end{figure}

We recall that $\S^3$ is a Lie group double cover of $\SO(3)$.
A uniform distribution on a compact Lie group is defined in terms
of the (normalized) Haar measure of the group,
see, e.g.,~\cite[p.~9]{Applebaum:2014:PCLG}. Such a measure is
invariant with respect to the left (or right) group multiplication
and is unique up to a constant multiple. For an introduction to the
Haar measure on Lie groups, see, e.g.,~\cite[p.~179-194]{Duistermaat:2012:LG}.
It is well-known that the Lebesgue measure $\dV_{\S^3}$ is a
bi-invariant (non-normalized) Haar measure for the Lie group
of unit quaternions.

\begin{theo}[Uniformly distributed random variables on $\SO(3)$]
  \label{theo:rejection_sampling_S3}
  Let $X_{\S^3}$ be a uniformly distributed random variable on $\S^3$ and
  $\pi: \S^3 \to \SO(3)$ the covering homorphism defined in~\eqref{eq:pi}.
  Then the random variable $X_{\SO(3)} \eqdef \pi \circ X_{\S^3}$ is uniformly
  distributed with respect to the (normalized) Haar measure on $\SO(3)$.
\end{theo}
\begin{proof}
  This is well-known, see, e.g.,~\cite{Shoemake:1992:URR}.
\end{proof}
Essentially, the covering homomorphism $\pi:\S^3 \to \SO(3)$ induces a
bi-invariant Riemannian metric on $\SO(3)$ via the pullback
$\scalprod{\cdot}{\cdot}_{\SO(3)} \eqdef (\pi^{-1})^*\scalprod{\cdot}{\cdot}_{\S^3}$.
Note that the bi-invariant metric on $\SO(3)$ is unique up to scalar
multiples since the Lie-algebra $\so(3)$ is simple.\footnote{Note that
  the Lie algebra $\so(n)$ is not simple for the exceptional dimensions
  $n = 2,4$.}
With respect to the pullback metric, $\pi$ is a local isometry. Further,
since $\pi$ is a covering map, the pullback of the invariant surface
volume measure on $\S^3$ given by $(\pi^{-1})^*\dV_{\S^3}$ induces an
invariant measure, i.e., a Haar measure, on $\SO(3)$.

On a sidenote, the use of Riemannian metrics and geodesics on (matrix)
Lie groups in applications is currently an active research area, since
the computational costs of geometric methods are no longer prohibitive.
For some interesting recent applications to strain measures
in mechanics, see, e.g.,~\cite{Neff:2015:GLS}. Another interesting recent
usecase for geodesics on the group of unit quaternions is the simulation
of eye movements, see~\cite{Novelia:2015:GSO}.

We now briefly describe the sampling strategies for the computational
validation.

\begin{rem}[Sampling the group of rotations $\SO(3)$]
  Based on the method described in~\theref{theo:rejection_sampling_S3},
  we have computed a set of samples $\mathcal{Q} \subset \S^3$
  consisting of $4.629.171$ uniformly distributed unit quaternions.
      \end{rem}

\begin{rem}[Sampling strategy for $F \in \GL^+(3)$]
  We have generated a stream of matrices with uniformly distributed
  coefficients $F_{ij} \in [-\frac{\sradmm}{2},\frac{\sradmm}{2}]$, $1 \leq i,j \leq 3$ and discarded all samples with $\det{F} < 0$.\footnote{We have also
    discarded matrices with non-simple singular values, but since these
    form a set of measure zero this case did never arise, as expected.}
  From the remaining samples, we have selected the first $1.000$ samples
  in $\domc_{\mu,\mu_c}$ and $\domn_{\mu,\mu_c}$, respectively, and
  collected them in two sets $\mathcal{F}^{\mathrm{C}}_{\mu,\mu_c}$
  and $\mathcal{F}^{\mathrm{NC}}_{\mu,\mu_c}$.
\end{rem}

\begin{rem}[Limitations of the sampling strategy for $F \in \GL^+(3)$]
  For performance reasons our sampling strategy takes our expectations
  into account right from the onset. This can be seen as a limitation.
  Further, our validation is inherently limited to compact subsets of
  $\GL^+(3)$. However, this particular strategy, heuristically produces
  a reasonable resolution for parameters $F \in \GL^+(3)$ in the vicinity
  of the branching condition $\sigma_1 + \sigma_2 = \sradmm$ (cf.
  our~\figref{fig:hat{beta}MC_NC}). Based on the predictions
  of the analysis of our proposed optimal Cosserat rotations presented
  in~\secref{sec:discussion}, this is without doubt the most interesting
  parameter sector.
\end{rem}

We are now finally in the position to expose our computational validation
strategy for the global optimality of the formula $\rpolar^\pm_{\mu,\mu_c}(F)$
stated in~\coref{cor:rpolar_formula}; cf. also~\remref{rem:rpolar_class}
and~\remref{rem:rpolar_nonclass} for a short review of the geometric
interpretation of the optimal Cosserat rotations. It is important to
note that the presented validation scheme is based on the lift
$$W^\sharp_{\mu,\mu_c}(q\,;F)\;\eqdef \wmm(\pi(q)\,;F)\;.$$

This formulation is based on the \emph{original} Cosserat-shear stretch
energy $\wmm(R\,;F)$, precisely as it appears in the statement
of~\probref{intro:prob_wmm}. Clearly, this allows to validate
the consistency of the simplifications leading us
to~\probref{prob:relative_rhat} in~\secref{sec:intro}.\footnote{Note that
  this also extends to our use of the parameter
  reduction~\cite[Lem.~2.2, p.~4]{Fischle:2015:OC2D}.}
This approach also implies that the image of the covering homomorphism
$\pi: \S^3 \to \SO(3)$ corresponds to an \emph{absolute} rotation
$\pi(q) = R$.

Let us now present our
\begin{compval}
  Let the sample sets $\mathcal{Q} \subset \S^3$ and
  $\mathcal{F}_{\mu,\mu_c} \eqdef \mathcal{F}_{\mu,\mu_c}^{\mathrm{C}} \cup \mathcal{F}_{\mu,\mu_c}^{\mathrm{NC}} \subset \GL^+(3)$ be as previously defined
  and set the numerical tolerance $\mathrm{tol} = 10^{-4}$.
  Then for all $\mu > 0$ and $\mu_c \geq 0$ (which we have tested)
  the following relation holds:
  \begin{equation}
    \forall F \in \mathcal{F}_{\mu,\mu_c}:\quad
    \pi\left(\argmin{q \in \mathcal{Q}}{W^\sharp_{\mu,\mu_c}(q\,;F)}\right)
    \quad=_{\mathrm{tol}}\quad \rpolar^\pm_{\mu,\mu_c}(F)\;,
  \end{equation}
  where $R_1 =_\mathrm{tol} R^\pm_2 \isequivto \min_{\pm}\hsnorm{R_1 - R^\pm_2} < \mathrm{tol}$, $R_1,R^\pm_2 \in \SO(3)$.
\end{compval}
The following procedure is equivalent, but more explicit. It also corresponds
more closely to our actual implementation:
\begin{equation}
  \begin{aligned}
    \forall F \in \mathcal{F}_{\mu,\mu_c}^{\mathrm{C}}:
    \quad
    \polar(F)^T \cdot \pi\left(\argmin{q \in \mathcal{Q}}{W^\sharp_{\mu,\mu_c}(q\,;F)}\right)
    &\quad=_{\mathrm{tol}}\quad \{\id\}
    \;,\\
    \forall F \in \mathcal{F}_{\mu,\mu_c}^{\mathrm{NC}}:\quad
    \polar(F)^T \cdot \pi\left(\argmin{q \in \mathcal{Q}}{W^\sharp_{\mu,\mu_c}(q\,;F)}\right)
    &\quad\equiv_\mathrm{tol}\quad [\hat{\beta}^\pm_{\mu,\mu_c}(F), q_3(F)]\;.
  \end{aligned}
\end{equation}
In order to clarify the meaning of the notation $\equiv_\mathrm{tol}$, let
$[\hat{\beta}^\pm_{\mu,\mu_c}(F), q_3(F)] \equiv R^\pm_2 \in \SO(3)$, then
$R_1 \equiv_\mathrm{tol} [\hat{\beta}^\pm_{\mu,\mu_c}(F), q_3(F)] \isequivto R_1 =_\mathrm{tol} R_2^\pm$; cf.~also~\remref{rem:rpolar_nonclass}.

\begin{rem}[Additional verification by a Riemannian Newton--scheme (W. M\"uller)]
  For some selected values of $F \in \GL^+(3)$, W. M\"uller (then at
  Karlsruhe Institute of Technology~\cite{Mueller09_diss}), has verified
  that the proposed formula $\rpolar^\pm_{\mu,\mu_c}(F)$ is a critical
  point for the Cosserat shear--stretch energy. He successfully
  approximated our proposed optimal Cosserat rotations
  up to machine accuracy by using a Riemannian Newton--scheme
  for the solution of the Euler--Lagrange equations on $\SO(3)$.
  Perturbations of the starting values of the Newton--iteration
  did not indicate the existence of alternative solutions realizing
  lower energy levels.
\end{rem}

In~\figref{fig:hat{beta}MC_NC}, we present multiple plots of the energy-minimizing
relative rotation angles $\hat{\beta}_{\mu, \mu_c}$ obtained by
stochastic (Monte Carlo) minimization. We show plots for different values
of ${\mu, \mu_c}$. A corresponding, in itself rather uninteresting, plot for
the classical limit case $(\mu,\mu_c) = (1,0)$ is depicted in~\figref{fig:hat{beta}MC_C} for direct comparison. Both figures match our expectations
raised by~\figref{fig:branchDiag} very well and the resolution does
improve with higher sample counts. It is instructive to compare these
figures with the optimal relative rotation angles for optimal \emph{planar}
Cosserat rotations presented in Part I of the present contribution,
see~\cite{Fischle:2015:OC2D}.

\begin{figure}
  \parbox{0.80\textwidth}{
    \begin{tikzpicture}
      \node (Pic) at (0,0)
            {\includegraphics[width=11cm]{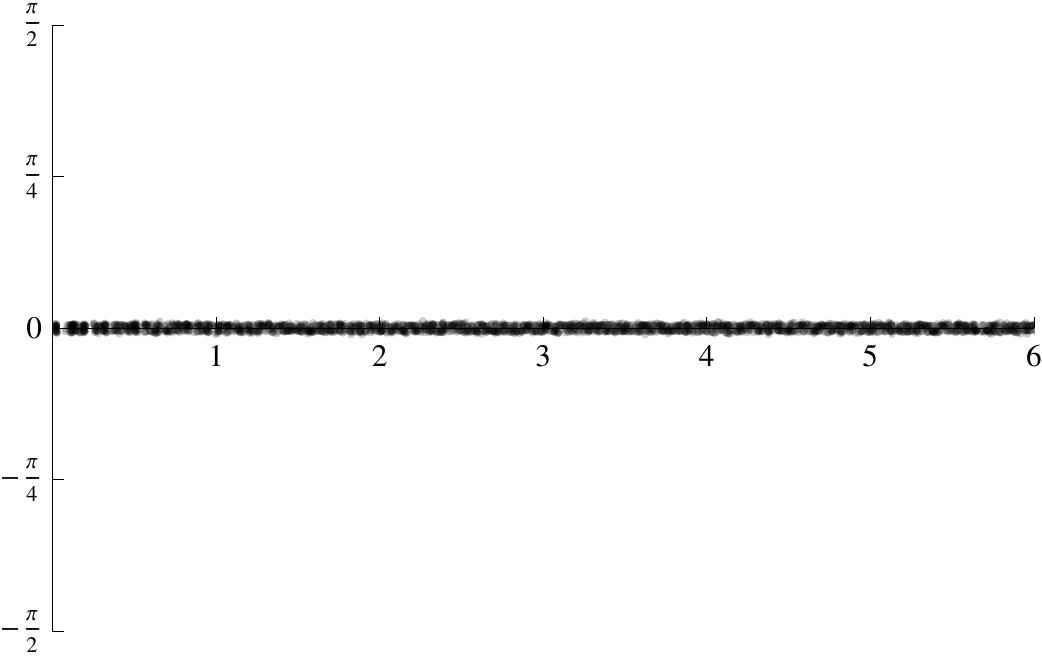}};
  \end{tikzpicture}}  \parbox{0.19\textwidth}{
    Parameters: $(\mu,\mu_c) = (1, 1)$\\
    $\rho_{1,1} = \infty$
  }
  \caption[Optimal classical relative rotation angles]{\label{fig:hat{beta}MC_C}    Optimal relative rotation angle $\hat{\beta}_{1,1}^{\rm MC}$ obtained from
    stochastic (Monte Carlo) minimization for the classical limit case
    $\mu = \mu_c = 1$. We observe that the relative rotation angle vanishes
    up to numerical accuracy, since the polar factor $\polar(F)$ is always
    optimal in perfect accordance with Grioli's theorem,
    see~\cite{Neff_Grioli14} and~\cite[Cor.~2.4,~p.~5]{Fischle:2015:OC2D}.
    More precisely, this corresponds to the prediction $\hat{\beta}^\pm_{1,1}(\sigma_1 + \sigma_2) = 0$. For multiple examples from the non-classical parameter
    range $\mu > \mu_c \geq 0$, see~\figref{fig:hat{beta}MC_NC} on page~\pageref{fig:hat{beta}MC_NC}.}
\end{figure}

\begin{figure}
  \parbox{0.8\textwidth}{
    \begin{tikzpicture}
      \node (Pic) at (0,0)
            {\includegraphics[width=11.0cm]{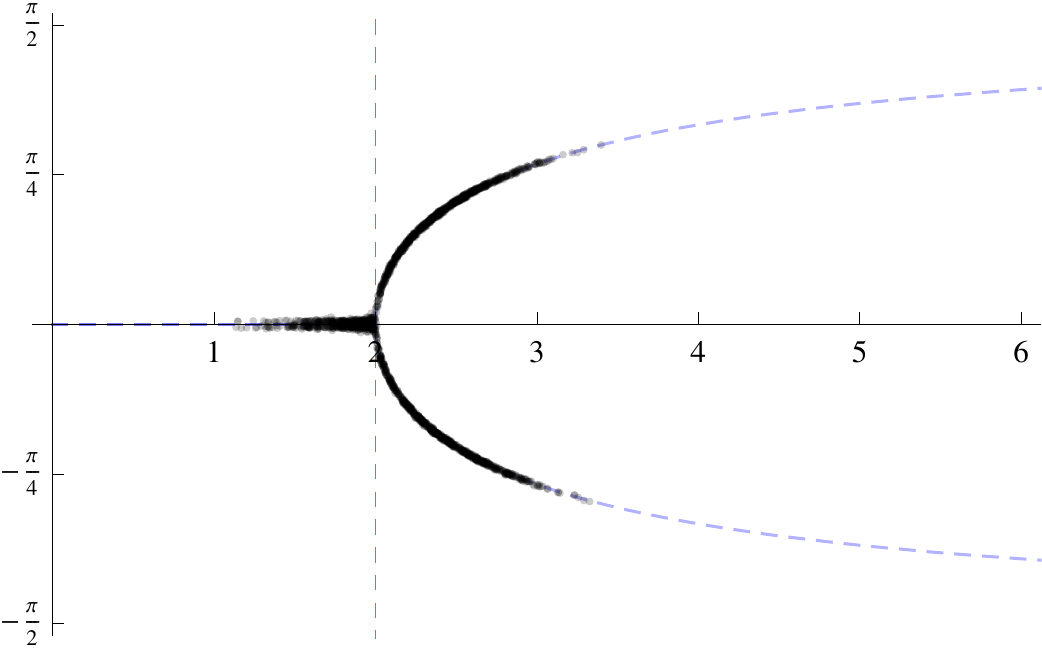}};
            \node (rho) at (-0.75, -2.5) {$\rho_{1,0} = 2$};
  \end{tikzpicture}}  \parbox{0.19\textwidth}{
    Parameters: $(\mu,\mu_c) = (1,0)$\\
    $\rho_{1,0} = 2$
  }
  \parbox{0.8\textwidth}{
    \begin{tikzpicture}
      \node (Pic) at (0,0)
            {\includegraphics[width=11.0cm]{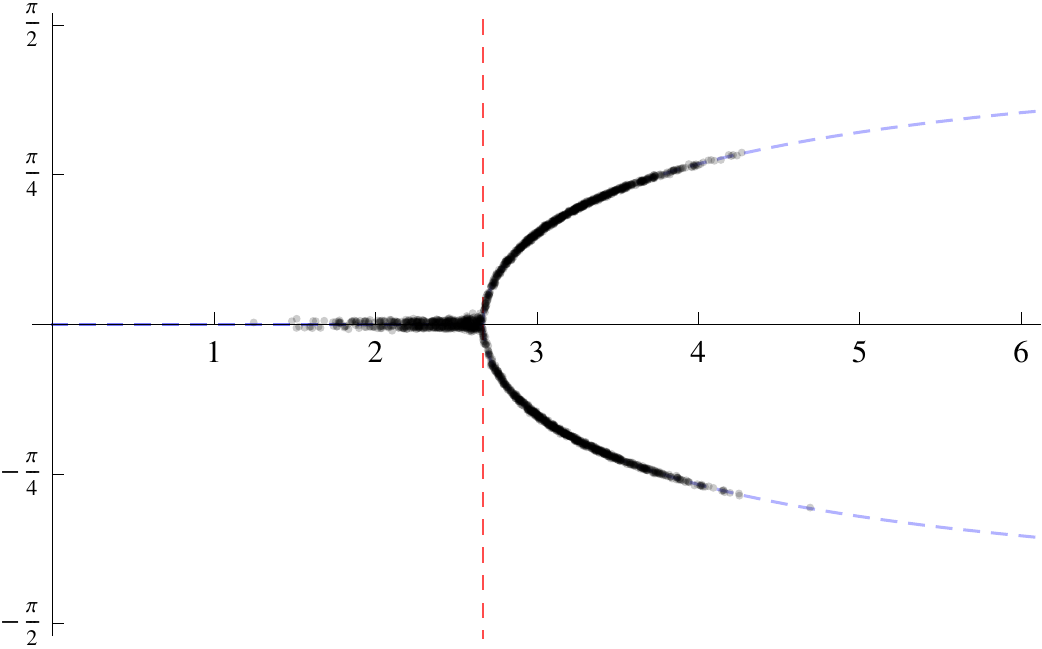}};
      \node (rho) at (0.50, -2.5) {$\rho_{1,\frac{1}{4}} = \frac{8}{3}$};
    \end{tikzpicture}}  \parbox{0.19\textwidth}{
    Parameters: $(\mu,\mu_c) = (1,\frac{1}{4})$\\
    $\rho_{1,\frac{1}{4}} = \frac{8}{3}$
  }
  \parbox{0.8\textwidth}{
    \begin{tikzpicture}
      \node (Pic) at (0,0)
            {\includegraphics[width=11.0cm]{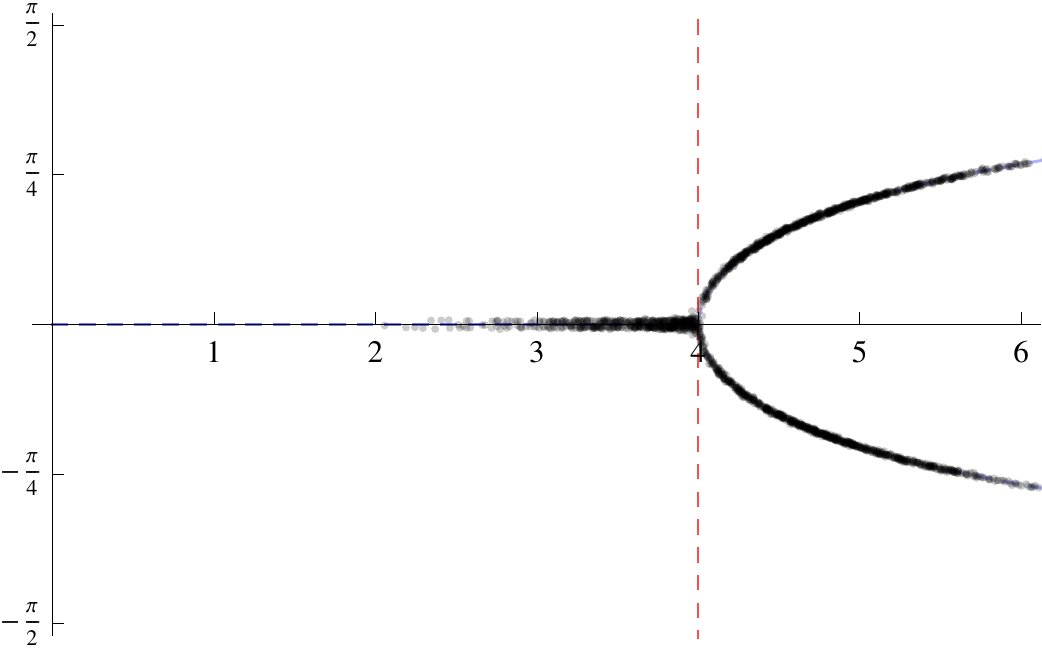}};
            \node (rho) at (2.75, -2.5) {$\rho_{1,\frac{1}{2}} = 4$};
  \end{tikzpicture}}  \parbox{0.19\textwidth}{
    Parameters: $(\mu,\mu_c) = (1, \frac{1}{2})$\\
    $\rho_{1,\frac{1}{2}} = 4$
  }
  \caption[Optimal non-classical relative rotation angles]{\label{fig:hat{beta}MC_NC}
    Optimal relative rotation angles $\hat{\beta}^{\rm MC}_{\mu,\mu_c}$ for
    multiple non-classical values $\mu > \mu_c \geq 0$. The angles are
    obtained by stochastic (Monte Carlo) minimization of $W_{\mu,\mu_c}(\mrot\,;F)$.
    The dashed blue curve shows the predicted value for
    $\hat{\beta}_{1,0}^\pm(\sigma_1 + \sigma_2)$ and
    the dashed red line marks the expected bifurcation point
    at $\rho_{\mu,\mu_c}$. For a direct comparison, we
    provide~\figref{fig:hat{beta}MC_C} on
    page~\pageref{fig:hat{beta}MC_C} which shows the classical
    limit case $(\mu,\mu_c) = (1,1)$; see also~\figref{fig:branchDiag}
    on page~\pageref{fig:branchDiag} for an illustration and a more
    precise description of the bifurcation behavior predicted by
    our proposed formula~$\rpolar^\pm_{\mu,\mu_c}(F)$.}
\end{figure}
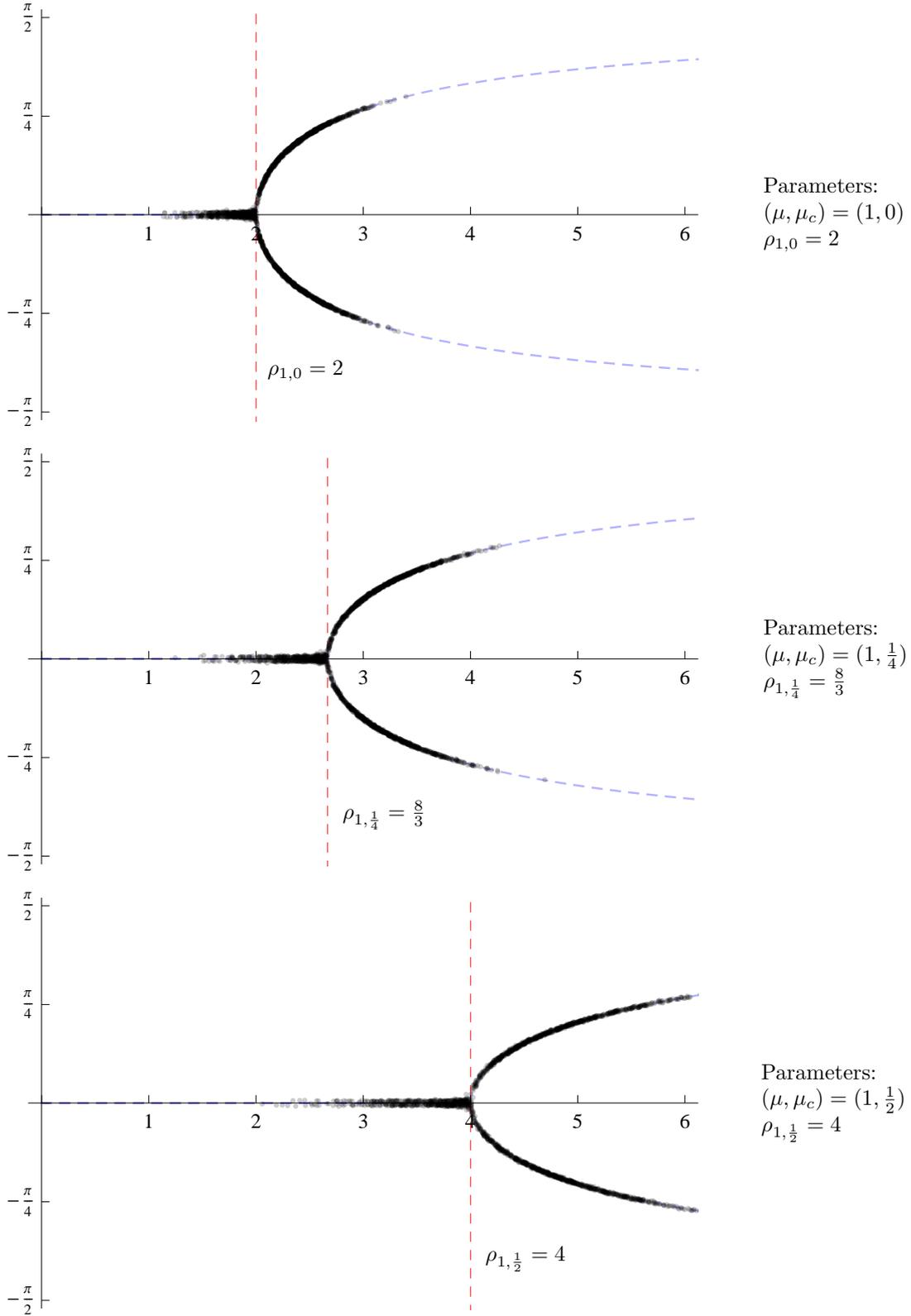
 \countres
\section{Conclusion}
\label{sec:conclusion}
The reduced Cosserat shear--stretch energy $\wmmred$ for which we have
finally obtained an explicit form in~\theref{theo:wmm_explicit} admits
an interesting abstract interpretation in mechanics. In order to reveal
this, let us first assume that the microrotations $\mrot$ are spatially
decoupled. This is the case when the length scale parameter $L_{\rm c}$
in the full Cosserat model, i.e., including a curvature energy
contribution, is extremely small or zero.
Let us furthermore assume that $\det{F} = 1$, i.e., that the amount of volume
distortion is negligible and that a specimen $\Omega$ of this material is
subjected to a given deformation $\varphi:\Omega \to \varphi(\Omega)$ with
deformation gradient $F \eqdef \Dphi \in \GL^+(3)$. Then the total reduced
Cosserat shear--stretch energy obtained by integration of the local density
given by
\begin{equation}
  \int\nolimits_\Omega \wmmred(F)\dV
  \;\eqdef\;
  \int\nolimits_\Omega \min_{R\,\in\,\SO(3)} \wmm(R\,;F)\dV
\end{equation}
corresponds precisely to the total energetic response which is generated
if the field of microrotations $\mrot$ in the specimen \emph{instantaneously}
aligns itself with the field of locally optimal Cosserat rotations
$\rpolar^\pm_{\mu,\mu_c}(\Dphi)$. It is important to observe that the the
field of optimal Cosserat rotations is \emph{purely} induced by the
deformation mapping $\varphi$ on which it depends by \emph{local}
energy minimization and does not otherwise depend on boundary
conditions, exterior forces, etc.

A Cosserat material which conforms to the previous description
can be nicely embedded into a classical framework due to G. Capriz,
the description of which is one of many shimmering pearls to be
found in the impressive body of his work on micropolar materials, see,
e.g.,~\cite{Capriz:1977:FSC,Capriz:1989:CWM}, and it is with delight
that we summarize it in a brief
\begin{rem}[Continua with latent microstructure in the sense of Capriz]
  \label{rem:capriz}
  In his paper~\cite[p.49]{Capriz85}, G. Capriz introduces the notion
  of a continuum with {\bf latent microstructure} as follows:
  \begin{quote}
    ``I say that the microstructure is latent when, though its effects
    are felt in the balance equations, all relevant quantities can
    be expressed in terms of geometric and kinematic quantities
    pertaining to apparent placements.''
  \end{quote}
  Capriz then gives a more precise definition of the properties a
  latent microstructure needs to satisfy. We shall only repeat the
  first two:
  ``There is no inertia connected with the microstructure.'', and,
  ``There are no exterior body actions on the microstructure.''
  In other words, a latent microstructure is coupled with a
  deformation $\varphi$ in an instantaneous way.
\end{rem}
The reduced Cosserat shear-stretch energy $\wmmred(F)$ can be considered
as the energetic answer of a medium with a rotational microstructure that
instantaneously reorganizes its field of microrotations $R: \Omega \to \SO(3)$
as an energy-minimizing $\rpolar^\pm_{\mu,\mu_c}(F)$-field. This is an
example for a latent microstructure in the sense of Capriz.

From a more general perspective, a Cosserat continuum can also be
considered as a special case of a so-called micromorphic model, see,
e.g.,~\cite{Neff_Forest_jel05,Neff_micromorphic_rse_05}
and~\cite{Matteo:2015:MMD}.  Let us, as before, set the length scale
parameter $L_{\rm c}$ governing the
curvature contribution to zero. We then observe that such an approach
\emph{always} leads to an algebraic side condition, in our case it is
given by the equation~\eqref{eq:EL_SO}, which replaces the partial
differential equation for the micro-distortion field. This is another,
more general, example of a continuum with latent microstructure in the
sense of Capriz, compare, e.g.,~\cite{Capriz:2000:PM}, due to G. Capriz
himself and also~\cite{Matteo:2015:MMD}.

Note that in~\cite{Demirkoparan:2014:HIB} and~\cite{Demirkoparan:2015:SSIB},
the authors -- who are apparently unaware of this established and
relatively straightforward interpretation -- have, in our opinion,
recently reintroduced the framework of materials with latent
microstructure due to G. Capriz for such micromorphic continuum
models under the new name of a hyperelastic material
with ``internal balance'' and an ``internally balanced solid'',
respectively.

We now continue our conclusion with some thoughts on possible
generalizations of our present results.

\begin{rem}[On generalizations to higher dimensions $n \geq 4$]
  Our solution approach is quite specifically tailored to
  dimension $n = 3$ since it relies on the covering of $\SO(3)$
  by the unit quaternions $\S^3 \subset \H$. It seems reasonable
  to assume that the particularly simple geometry of $\S^3$
  lies at the root of the explicit solvability of the Euler--Lagrange
  equations. The so-called Sphere Theorem states that the only
  spheres that admit a connected compact Lie group structure
  are $\S^1$ and $\S^3$, see, e.g.,~\cite[p.289]{Hofmann:2006:SCG}.
  Thus, for $n > 3$, there is no hope at all to recover the particularly
  simple constellation we have quite successfully exploited here.
  Still, there is a generalization of the unit quaternions, namely
  the so-called spin groups $\mathrm{Spin}(n)$. These groups are
  two-fold covers
  of $\SO(n)$ and closely related to Clifford algebras, see,
  e.g.,~\cite{Lawson:1989:SG} and~\cite{Doran:2003:GAP}. In principle,
  such techniques might be appropriate for a generalization of our
  present results to higher dimensions, but they are out of reach
  for us.
\end{rem}

Although our exact solution approach does not generalize
to higher dimensions, it seems obvious that the
reduced~\probref{prob:relative_rhat} is a very good starting
point for the solution of~\probref{intro:prob_wmm} in
dimensions $n \geq 4$. Given this particular form, it seems
very likely that the minimizers in higher dimensions can also
be characterized in terms of the eigenvectors of $U = QDQ^T$
and the singular values $\sigma_i$, $1 \leq i \leq n$, of
$F \in \GL^+(n)$. Similar to the rather simplistic random
sampling strategy we have employed here,
it might certainly be worthwhile to carry out an initial
investigation based on a suitable Monte Carlo random sampling
approach which is suitable for higher dimensions, see,
e.g.,~\cite{Leon:2006:SMR}. On a related note, we have to
dampen expectations regarding extensions to anisotropic
formulations. These seem to be completely out of reach, since a
reduction to a formulation in singular values is then
impossible, see~\cite{Neff_Muench_transverse_cosserat08}
and~\cite{Pau:2012:BMMP}.

Another interesting question which is raised by our findings
is whether the maximal mean planar stretch
and strain ``measures'', i.e., $\ump(F)$ and $\smp(F)$, as defined
in~\deref{defi:mmpss} -- which appear to be such natural concepts
in our particular context -- are just artifacts of our derivation.
The same holds for the plane of maximal
strain $\mathrm{P}^{\rm mp}(F)$ introduced in~\deref{defi:pms}. Are there
real-world materials or material models which can be precisely or at least
approximately characterized by, e.g., slip in the plane of maximal
strain $\mathrm{P}^{\rm mp}(F)$? Currently, we are not aware of any such
materials or models.

In good hope that the presented mechanisms and computational strategies
will be at least helpful for the derivation of closed-form solutions
for~\probref{intro:prob_wmm} in dimensions $n \geq 3$ and that these
will match our proposed formula $\rpolar^\pm_{\mu,\mu_c}(F)$ presented
in~\coref{cor:rpolar_formula} for $n = 3$, we conclude
our present contribution with a last
\begin{rem}[Final remark]
  As regards suitable values of the Cosserat couple modulus $\mu_c \geq 0$,
  our development shows clearly that there are ultimately only 3 values of
  particular interest, namely
  \begin{equation*}
    \mu_c \;=\; 0\;,\quad\quad \mu_c \;=\; \mu\;,\quad\quad \text{and} \quad\quad \mu_c \;=\; +\infty\;.\tag*{$\blacksquare$}
  \end{equation*}
\end{rem}
 \countres

\addcontentsline{toc}{section}{References}
\bibliographystyle{plain}

{\footnotesize
  \setlength{\bibsep}{1.25pt}
  \bibliography{./fine-2015-cosserat3d}
}
\FloatBarrier

\newpage
\appendix
\section{Appendix (list of critical points)}
\label{sec:appendix}
{\footnotesize
We now detail our computer assisted strategy for the computation of the
critical points for the Lagrange function $\widehat{L}_{1,0}$. We recall
that the Euler--Lagrange equations simplify considerably if one
(or more) of the quaternion coefficients $w,x,y$ or $z$ vanishes.
This is reflected in the solution set computed by \mathematica.

We recall our shorthand notation for sums and differences of singular
values for parameters $F \in \GL^+(n)$ introduced in~\secref{sec:minimization}:
\begin{equation*}
  s_{ij} \eqdef \sigma_i + \sigma_j\quad\text{and}\quad
  d_{ij} \eqdef \sigma_i - \sigma_j\;,\quad i,j = 1,2,3\;.
\end{equation*}

\subsection{Computation of critical points of the Lagrange function $\widehat{L}_{1,0}$}
In order to solve the Euler--Lagrange equations in quaternion
representation~\eqref{eq:EL_quat}, we have used the \texttt{Reduce}
command in \mathematica. This command returns 130 critical points for
the Lagrange function $\widehat{L}_{1,0}$.\footnote{Note that the
  \texttt{Reduce} function in \mathematica~is supposedly \emph{guaranteed}
  to compute a \emph{complete} solution set. This is not
  necessarily the case for the \texttt{Solve}-command which generates
  only \emph{generic} solutions, see~\cite{Mathematica10:2015:GNG}.}
The solution set is enlarged by particular solutions satisfying certain
algebraic relations among the parameters $\sigma_i$, $i = 1,\ldots,3$.
In order to exclude these cases, we have made the following assumptions
\begin{equation}
  \sigma_i \neq 0\;,\quad
  s_{ij} \eqdef \sigma_i + \sigma_j \neq 0\;,
  \quad \text{and}\quad
  d_{ij} \eqdef \sigma_i - \sigma_j \neq 0\;,\quad i,j = 1,2,3,\; i \neq j\;.\label{eq:app:assumptions}
\end{equation}
These assumptions can be passed to the \texttt{Reduce} command in the form
of a so-called \texttt{AssumptionList}. This is a standard procedure. With these
assumptions \mathematica~successfully symbolically reduces the full solution
set comprised of 130 critical points to 32 critical points. Note that a strict
ordering of the parameters $\sigma_1 > \sigma_2 > \sigma_3 > 0$ implies that
all of these assumptions in~\eqref{eq:app:assumptions} are satisfied.
We shall refer to the 32 branches so obtained as the set of \emph{generic}
solutions, since they coincide with the output of the \texttt{Solve} command.
Due to this procedure, the possibility of multiple singular values of $F$
and the degeneracy of $\det{F} = 0$ are explicitly excluded.

\begin{rem}[Completeness of the solution set]
  We want to stress that we present manually refined results obtained
  via the computational algebra system (CAS) \mathematica. Currently,
  we cannot \emph{prove} that our solution set is complete, but this is
  quite probably the case as our extensive validation shows.
\end{rem}

It might be an interesting challenge for an expert in (computational)
algebraic geometry to prove that the presented list of \emph{generic}
solutions to the polynomial
system~\eqref{eq:EL_quat} is in fact \emph{complete}.\footnote{Note that an attempt
  to compute either a Gr\"obner basis or a primary decomposition for the
  Euler--Lagrange equations~\eqref{eq:EL_quat} using the
  CAS~\texttt{Singular}~\cite{Singular4} with competent assistance by
  R. Vollmert and L. Kastner~\cite{Vollmert:2015:OPD,Vollmert:2012:SDT}
  have failed (to finish within a day). The parameter-dependent
  polynomial system might be non-trivial to solve by computer
  algebra. It seems to us that \mathematica~automatically carries out
  the case distinctions $w = 0 \lor w \neq 0$, $x = 0 \lor x \neq 0$, etc.,
  since a computation of a Gr\"obner basis does not seem to terminate
  either. These case distinctions simplify the Euler--Lagrange equations,
  but we can just speculate here.}

Note that the computations can also be based on the alternative continuation
$\pi':\H \to \SO(3)$ which produces slightly different Euler--Lagrange
equations but yields the same full and generic solution sets.

In order to describe the set of generic solutions, we introduce two
auxiliary coefficient functions
\begin{equation}
  c_A(t) \eqdef \sqrt{\frac{1}{2} + \frac{1}{t}}\;,
  \quad\text{and}\quad
  c_B(t) \eqdef \sqrt{\frac{1}{2} - \frac{1}{t}}\;.\label{eq:app:cfunctions}
\end{equation}
It is important to note that the function $c_B(t)$ is only real-valued
for $t \in [2,\infty)$. The generic solutions are critical tuples
of \emph{complex} variables $(w,x,y,z,\lambda) \in \C^5$, i.e., the
solutions obtained by \mathematica~are complex-valued. It turns out
that the critical tuples are real-valued if and only if the radicands
in~\eqref{eq:app:cfunctions} are non-negative which allows to deduce
corresponding domains of definition for the critical branches.

\begin{rem}[Covering symmetry]
For any critical pair $(\hat{q},\lambda)$, the pair $(-\hat{q}, \lambda)$
is also a critical point for the Lagrange function $\widehat{L}_{1,0}$.
The associated reduced energy levels are identical.
\end{rem}
Due to this antipodal symmetry, it suffices to restrict the following
presentation to $16$ solutions. For their exposition, we have decided
to introduce three categories.

\subsubsection{Type I (of critical points)}

Characterization: Precisely one of the coefficients $w,x,y$ and $z$ of $\hat{q}(w,x,y,z)$ is non-zero. The solutions are independent of the parameters $\sigma_i$, $i = 1,2,3$.

These correspond to the following energy-minmizing \emph{relative} rotations:
\begin{equation}
\begin{matrix*}[l]
  \hat{q}_{{\mathrm{I}},1} &\eqdef &(1,\, 0,\, 0,\, 0) &\equiv &\left[0,\, \quad\text{any}\quad\right] &\equiv &\id_3                 &= &Q_1\;,\\
  \hat{q}_{{\mathrm{I}},2} &\eqdef &(0,\, 1,\, 0,\, 0) &\equiv &\left[\pi, (1,\,0,\,0)\right] &\equiv &\diag(\phantom{-}1,\,-1,\,-1)  &= &Q_2\;,\\
  \hat{q}_{{\mathrm{I}},3} &\eqdef &(0,\, 0,\, 1,\, 0) &\equiv &\left[\pi, (0,\,1,\,0)\right] &\equiv &\diag(-1,\,\phantom{-}1,\,-1)  &= &Q_3\;,\\
  \hat{q}_{{\mathrm{I}},4} &\eqdef &(0,\, 0,\, 0,\, 1) &\equiv &\left[\pi, (0,\,0,\,1)\right]  &\equiv &\diag(-1,\,-1,\,\phantom{-}1) &= &Q_4\;.
\end{matrix*}
\end{equation}
Note that the relative rotation represented by
$\hat{q}_{{\mathrm{I}},1} \equiv \id_3$
corresponds to the polar factor $\polar(F)$.

For the Lagrange multiplier $\lambda$, we obtain the associated critical values
\begin{equation}
  \lambda_{\mathrm{I},1} \eqdef 0\;,\quad
  \lambda_{\mathrm{I},2} \eqdef \sigma_2^2 + \sigma_3^2 + 4\, s_{23}\;,\quad
  \lambda_{\mathrm{I},3} \eqdef \sigma_3^2 + \sigma_1^2 + 4\, s_{31}\;,
  \quad\text{and}\quad
  \lambda_{\mathrm{I},4} \eqdef \sigma_1^2 + \sigma_2^2 + 4\, s_{12}\;.
\end{equation}

The realized energy levels of the lifted energy are given by:
\begin{equation}
  \begin{aligned}
  \widehat{W}_{1,0}^\sharp(\hat{q}_{{\mathrm{I}},1}\,;D) &= (\sigma_1 - 1)^2 + (\sigma_2 - 1)^2 + (\sigma_3 - 1)^2\;,\\
  \widehat{W}_{1,0}^\sharp(\hat{q}_{{\mathrm{I}},2}\,;D) &= (\sigma_1 + 1)^2 + (\sigma_2 + 1)^2 + (\sigma_3 - 1)^2\;,\\
  \widehat{W}_{1,0}^\sharp(\hat{q}_{{\mathrm{I}},3}\,;D) &= (\sigma_1 - 1)^2 + (\sigma_2 + 1)^2 + (\sigma_3 + 1)^2\;,\\
  \widehat{W}_{1,0}^\sharp(\hat{q}_{{\mathrm{I}},4}\,;D) &= (\sigma_1 + 1)^2 + (\sigma_2 - 1)^2 + (\sigma_3 + 1)^2\;.
  \end{aligned}
\end{equation}

The solutions of the first type are globally defined.

\subsubsection{Type II (of critical points)}

Characterization: Precisely two of the coefficients $x,y,z$ of 
                  $\hat{q}(w,x,y,z)$
                  vanish. Further, the solution only depends on 
                  pairwise \emph{sums} of the singular values 
                  $\sigma_i$, $i = 1,2,3$.

These correspond to the following energy-minmizing \emph{relative} rotations:
\begin{equation}
\hspace{-0.5cm}
\begin{matrix*}[l]
    \hat{q}^\pm_{{\mathrm{II}},1} \eqdef \left( c_A(s_{12}),\, 0,\, 0, \pm c_B(s_{12})\right) &\equiv&   \left[\pm\arccos(2/s_{12}), (0,\,0,\,1)\right]\;,\\[.5em]
    \hat{q}^\pm_{{\mathrm{II}},2} \eqdef \left( c_A(s_{23}),\, \pm c_B(s_{23}),\, 0,\, 0\right) &\equiv& \left[\pm\arccos(2/s_{23}), (1,\,0,\,0)\right]\;,\\[.5em]
    \hat{q}^\pm_{{\mathrm{II}},3} \eqdef \left( c_A(s_{31}),\, 0,\, \pm c_B(s_{31}),\, 0\right) &\equiv& \left[\pm\arccos(2/s_{31}), (0,\,1,\,0)\right]\;.
  \end{matrix*}
\end{equation}

For the Lagrange multiplier $\lambda$, we obtain the associated critical values
\begin{equation}
  \lambda^\pm_{\mathrm{II},1} \eqdef d_{12}^2 \left(\frac{s_{12} - 2}{s_{12}}\right)\;,\quad
  \lambda^\pm_{\mathrm{II},2} \eqdef d_{23}^2 \left(\frac{s_{23} - 2}{s_{23}}\right)\;,
  \quad\text{and}\quad
  \lambda^\pm_{\mathrm{II},3} \eqdef d_{31}^2 \left(\frac{s_{31} - 2}{s_{31}}\right)\;.
\end{equation}

The realized energy levels of the lifted energy are given by:
\begin{equation}
  \begin{aligned}
    \widehat{W}_{1,0}^\sharp(\hat{q}^\pm_{{\mathrm{II}},1}\,;D) = \frac{1}{2}(\sigma_1 - \sigma_2)^2 + (\sigma_3 - 1)^2\;,\\
    \widehat{W}_{1,0}^\sharp(\hat{q}^\pm_{{\mathrm{II}},2}\,;D) = \frac{1}{2}(\sigma_2 - \sigma_3)^2 + (\sigma_1 - 1)^2\;,\\
    \widehat{W}_{1,0}^\sharp(\hat{q}^\pm_{{\mathrm{II}},3}\,;D) = \frac{1}{2}(\sigma_3 - \sigma_1)^2 + (\sigma_2 - 1)^2\;.\\
  \end{aligned}
\end{equation}

The solutions of the second type are defined for
$s_{ij} \in [2,\infty)$, $i,j = 1,2,3$.

\subsubsection{Type III (of critical points)}

Characterization: The coefficient $w$ vanishes together with exactly one of the
                  remaining coefficients $x,y,z$ of $\hat{q}(w,x,y,z)$. 
                  Further, the solution only depends on pairwise
                  \emph{differences} of the singular values 
                  $\sigma_i$, $i = 1,2,3$.

These correspond to the following energy-minmizing \emph{relative} rotations:
\begin{equation}
\begin{matrix*}[l]
  \hat{q}^\pm_{{\mathrm{III}},1} \eqdef \left( 0 ,\, c_A(d_{12}) ,\, \pm c_B(d_{12}) ,\, 0 \right)                   &\equiv&  \left[\pi, \left(c_A(d_{12})^{-\frac{1}{2}},\, \pm c_B(d_{12})^{\frac{1}{2}},\,0\right)\right]\;,\\[.75em]
  \hat{q}^\pm_{{\mathrm{III}},2} \eqdef \left( 0 ,\, 0               ,\, c_A(d_{23})     ,\, \pm c_B(d_{23}) \right) &\equiv&  \left[\pi, \left(0,\,c_A(d_{23})^{-\frac{1}{2}},\, \pm c_B(d_{23})^{\frac{1}{2}}\right)\right]\;,\\[.75em]
  \hat{q}^\pm_{{\mathrm{III}},3} \eqdef \left( 0 ,\, c_A(-d_{31}) ,\, 0                   ,\, \pm c_B(-d_{31}) \right) &\equiv&  \left[\pi, \left(c_A(-d_{31})^{-\frac{1}{2}},\, 0,\, \pm c_B(-d_{31})^{\frac{1}{2}}\right)\right]\;.
\end{matrix*}
\end{equation}

For the Lagrange multiplier $\lambda$, we obtain the associated critical values
\begin{equation}
\begin{aligned}
  \lambda^\pm_{\mathrm{III},1} \quad&\eqdef\quad 4\, \sigma_3 (1 + \sigma_3) + (s_{23} - 2) s_{23}\;,\\
  \lambda^\pm_{\mathrm{III},2} \quad&\eqdef\quad 4\, \sigma_1 (1 + \sigma_1) + (s_{12} - 2) s_{12}\;,\\
  \lambda^\pm_{\mathrm{III},3} \quad&\eqdef\quad 4\, \sigma_2 (1 + \sigma_2) + (s_{31} - 2) s_{31}\;.
\end{aligned}
\end{equation}

The realized energy levels of the lifted energy are given by
\begin{equation}
\begin{aligned}
  \widehat{W}_{1,0}^\sharp(\hat{q}^\pm_{{\mathrm{III}},1}\,;D) \quad&=\quad\frac{1}{2}(\sigma_1 + \sigma_2)^2 + (\sigma_3 + 1)^2\;,\\
  \widehat{W}_{1,0}^\sharp(\hat{q}^\pm_{{\mathrm{III}},2}\,;D) \quad&=\quad\frac{1}{2}(\sigma_2 + \sigma_3)^2 + (\sigma_1 + 1)^2\;,\\
  \widehat{W}_{1,0}^\sharp(\hat{q}^\pm_{{\mathrm{III}},3}\,;D) \quad&=\quad\frac{1}{2}(\sigma_3 + \sigma_1)^2 + (\sigma_2 + 1)^2\;.
\end{aligned}
\end{equation}

The solutions of the third type are defined for $d_{ij} \in [2,\infty)$,
$i < j$, $i,j = 1,2,3$.

\begin{rem}[On second order conditions]
  We have succeeded to compute the signs of the principal minors of the
  so-called bordered Hessian
  $\mathrm{H}_{(\hat{q},\lambda)}\widehat{L}_{1,0}(\hat{q},\lambda\,;D)$
  evaluated at the previously presented $32$ critical points. This allows
  to carry out an analysis of the second order conditions
  for local constrained extrema based on an analysis of the sign changes
  of the principal minors of the bordered Hessian. The procedure
  in the constrained case is different but similar to the well-known
  procedure in the unconstrained case~\cite{Debreu:1952:DSQF,Hestenes:1975:OTF}.
\end{rem}
} \countres

\end{document}